\documentclass[12pt]{amsart}

\usepackage{amssymb, enumerate, mdwlist}

\evensidemargin\paperwidth
\advance\evensidemargin-\textwidth
\advance\oddsidemargin-1in
\evensidemargin\oddsidemargin

\topmargin\paperheight
\advance\topmargin-\textheight
\advance\topmargin-1in

\newcommand{\B}{\mathcal{B}}
\newcommand{\C}{\mathbb{C}}

\newcommand{\F}{\mathcal{F}}
\newcommand{\Hil}{\mathcal{H}}
\newcommand{\N}{\mathbb{N}}
\newcommand{\R}{\mathbb{R}}
\newcommand{\T}{\mathbb{T}}
\newcommand{\Z}{\mathbb{Z}}

\theoremstyle{plain}
\newtheorem{theorem}{Theorem}[section]
\theoremstyle{plain}
\newtheorem{proposition}[theorem]{Proposition}
\theoremstyle{plain}
\newtheorem{lemma}[theorem]{Lemma}
\theoremstyle{plain}
\newtheorem{corollary}[theorem]{Corollary}
\theoremstyle{plain}

\theoremstyle{plain}
\newtheorem{definition}[theorem]{Definition}
\theoremstyle{plain}

\theoremstyle{remark}
\newtheorem{remark}[theorem]{Remark}
\theoremstyle{remark}
\newtheorem{example}[theorem]{Example}
\theoremstyle{remark}

\title
[Intertwiners between $1$-dim QWs]
{Intertwining operators between one-dimensional homogeneous quantum walks}

\author{Hiroki Sako}
\address
{Faculty of Engineering, Niigata University, Nishi-ku, Niigata 950-2181, Japan}
\email
{sako@eng.niigata-u.ac.jp}

\begin{document}
\maketitle

\begin{abstract}
The subject of this paper is 
a kind of dynamical systems called {\it quantum walks}.
We study one-dimensional homogeneous analytic quantum walks $U$.
We explain how to identify the space of all the uniform
intertwining operators between these walks.
We can also determine 
whether $U$ can be realized by a (not necessarily homogeneous)
continuous-time uniform quantum walk on $\Z$.
Several examples of quantum walks, which can not be realized by continuous-time uniform quantum walks, are presented.
The $4$-state Grover walk is one of them.
Before stating the main theorems, we clarify the definition of one-dimensional quantum walks.
For the first half of this paper, we study basic properties of one-dimensional quantum walks, which are not necessarily homogeneous.
An equivalence relation between quantum walks called {\it similarity} is also introduced.
This allows us to manipulate quantum walks in a flexible manner.
\end{abstract}

\section{introduction}

{\it Quantum walka} are dynamical systems related to quantum physics.
Many researchers study this subject in a number of frameworks.
They commonly use a pair of a Hilbert space $\Hil$ and a unitary operator $U$ on $\Hil$. The Hilbert space is associated to some space $X$. The space $\Hil$ is given by $\ell_2(X)$, $L^2(X)$, or their amplification.


There are two families of quantum walks. One is the family of discrete-time quantum walks.
These walks give unitary representations $(U^t)_{t \in \Z}$ of the integer group $\Z$.
We can regard the integer $t$ as the number of steps of some procedure.
The other is the family of continuous-time quantum walks.
Such a walk gives a unitary representation $(\exp(i t H))_{t \in \R}$ of the real group $\R$.
We can regard the real number $t$ as the flow of time.

For a discrete-time quantum walk $(U^t)_{t \in \Z}$, does there exist a {\it good} continuous-time quantum walk $(\exp(i t H))_{t \in \R}$ satisfying $\exp(i X) = U$?
A related open problem is proposed in \cite{Ambainis}.
For every unitary operator $U$, there exists a self-adjoint operator $H$ such that $\exp(i H) = U$. However, $X$ is not necessarily a {\it good} operator.
The unitary operator $\exp(i t H)$ ignores the base space $X$. Namely, there might exist unit vectors $\xi$ and $\eta$ such that the support of $\eta$ in $X$ is distant from that of $\xi$, and that the transition probability $|\langle \exp(i t H) \xi, \eta \rangle|^2$ is not small.
This means that the dynamical system by $(\exp(i t H))_{t \in \R}$ moves unit vectors too fast.
Therefore, we eliminate such a pathological walk and concentrate on walks satisfying some regularity. In this paper, we consider three kinds of regularity called {\it uniformity}, {\it smoothness}, or {\it analyticity} for operators on $\Hil$. Uniformity is the weakest, and analyticity is the strongest.
To the best of the author's knowledge, 
all the known examples of one-dimensional quantum walks are analytic.

In this paper, we consider the case that the space $X$ is the integer group $\Z$, the local degree of freedom of $\Hil$ is finite, and $(U^t)_{t \in \Z}$ is a discrete-time homogeneous {\it analytic} quantum walk. 
We determine whether $(U^t)_{t \in \Z}$ is realized by a continuous-time {\it uniform} quantum walk in Theorem \ref{theorem: continuous-time}. 
To show this theorem, in Subsection \ref{subsection: uniform intertwiner}, we determine
the space of uniform intertwining operator between two one-dimensional homogeneous analytic quantum walks.

Before stating the main theorem (Theorem \ref{theorem: continuous-time}), we need to clarify the definition of one-dimensional quantum walks and regularity for operators on $\Hil$ in Section \ref{section: one-dimensional QW} (Definitions \ref{definition: one-dimensional QW}, \ref{definition: regularity}, \ref{definition: regularity for QW}).
Many results in Section \ref{section: one-dimensional QW} can be applied to general one-dimensional quantum walks, which are not necessarily of finite degree of freedom.
We also propose a new equivalence relation between one-dimensional quantum walks called {\it similarity}.
This new notion allows us to treat quantum walks in a flexible manner.
Similar walks have the same asymptotic behavior (Theorem \ref{theorem: similarity and limit distribution}).

For the argument in this paper, 
we need a structure theorem on one-dimensional homogeneous quantum walks in \cite{SaigoSako}.
A concise abstract of the paper \cite{SaigoSako} is described in Subsection \ref{subsection: review of SS}.
Readers who wants to concretely understand the contents of this paper are recommended to read examples, skipping lemmas and propositions.
Among several examples, Example \ref{example: 4-state Grover walk} and Example \ref{example: 3-state Grover walk} introduce the $4$-state Grover walk and the $3$-state Grover walk.
Example \ref{example: 3-state Grover walk is continuous-time} shows that the $3$-state Grover walk can be realized by a continuous-time quantum walk, while
Example \ref{example: 4-state Grover walk is not continuous-time} shows that the $4$-state Grover walk can not.

\section{Definitions and basic properties of $1$-dimensional quantum walks}
\label{section: one-dimensional QW} 

We construct a general framework for one-dimensional quantum walks as follows.

\begin{definition}\label{definition: one-dimensional QW}
One-dimensional {\rm discrete-time} quantum walk is a triplet $(\mathcal{H}$, $(U^{t})_{t \in \Z}$, $D)$ of 
\begin{itemize}
\item
a Hilbert space $\mathcal{H}$,
\item
a unitary representation $(U^t)_{t \in \Z}$ of $\Z$ on $\mathcal{H}$, 
\item
and a self-adjoint operator $D$. (In most cases, $D$ is unbounded.)
\end{itemize}
We call $U = U^1$ the generator of the quantum walk.
\end{definition}

\begin{definition}
One-dimensional {\rm continuous-time} quantum walk is a triplet $(\mathcal{H}$, $(U^{(t)})_{t \in \R}$, $D)$ of 
\begin{itemize}
\item
a Hilbert space $\mathcal{H}$,
\item
a one-parameter group $(U^{(t)})_{t \in \R}$ of unitary operators on $\mathcal{H}$ which is continuous with respect to the strong operator topology, 
\item
and a self-adjoint operator $D$. (In most cases, $D$ is unbounded.)
\end{itemize}
The self-adjoint operator $\lim_{t \to 0} \frac{U^{(t)} - 1}{it}$ is called the generator of the quantum walk. 
\end{definition}

For the rest of this paper, we concentrate on one-dimensional quantum walks, so we simply call them {\it quantum walk}.

In most of preceding research, quantum walks are regarded as a dynamical system on some geometric space.
To fit the quantum walks defined above, we have only to define the operator $D$ as the observable of position on a one-dimensional space.
However, the above definition allows more flexible interpretations of quantum walks.
The self-adjoint operator $D$ can be other observables such as the momentum operator.

\subsection{Regularity on quantum walks}

By physical requirement,
we often assume a kind of regularity for operators such as {\it uniformity}, {\it smoothness},
 or {\it analyticity}.
Note that for a map $f$ from the real line or a complex domain to a Banach space $B$, we can define differentiability on $f$ using the limit $\lim_{\Delta x \to 0} \frac{f(x + \Delta x) - f(x)}{\Delta x}$ in norm.

\begin{definition}\label{definition: regularity}
Let $D$ be a self-adjoint operator on $\Hil$ and let $U$ be a bounded operator on $\Hil$.
\begin{itemize}
\item
The operator $U$ is said to be {\rm uniform} with respect to $D$, if
the map
\[\R \ni k \mapsto e^{i k D} U e^{-i k D} \in \mathcal{B}(\mathcal{H})\]
is continuous with respect to the norm topology.
This condition implies that $k \mapsto e^{i k D} U e^{-i k D}$ is uniformly continuous.
\item
The operator $U$ is said to be {\rm smooth} or {\rm in the ${\rm C}^\infty$-class} with respect to $D$, if
the map
\[\R \ni k \mapsto e^{i k D} U e^{-i k D} \in \mathcal{B}(\mathcal{H})\]
is a smooth mapping with respect to the variable $k \in \R$.
\item
The operator $U$ is said to be {\rm analytic} with respect to $D$, if
there exists a holomorphic extension of the map 
\[\R \ni k \mapsto e^{i k D} U e^{-i k D} \in \mathcal{B}(\mathcal{H})\]
defined on a domain of the form $\{\kappa \in \C | -\delta < {\rm Im}(\kappa) < \delta \}$.
\end{itemize}
\end{definition}

These conditions related to transition probability in quantum mechanics.
Consider the case that the spectrum of $D$ stands for position of some particle and that $U$ corresponds to some dynamical system.
Let $E(\cdot)$ be the spectral measure of $D$.
Let $\xi$ and $\eta$ be unit vectors in $\Hil$.
Suppose that the support of the measure $\langle E(\cdot) \xi, \xi \rangle$ is distant from that of  $\langle E(\cdot) \eta, \eta \rangle$.
The conditions on regularity of $U$ mean that the matrix coefficient $\langle U \xi, \eta \rangle$ is small, if the support of $\xi$ with respect to the spectral decomposition of $D$ is distant from that of $\eta$. 
See \cite[Definition 3.1]{SaigoSako} and \cite[Lemma 4.1]{SaigoSako}.
See also Proposition \ref{proposition: converging to an atom}.
Among the three conditions, uniformity is the weakest, and analyticity is the strongest.
The space of all the uniform operators forms a C$^*$-algebra.
The space of all the smooth operators forms a $*$-subalgebra.
The space of all the analytic operators also forms a $*$-subalgebra.

The main subject of this paper is {\it uniform} intertwiner between two homogeneous discrete-time {\it analytic} quantum walks.

If the operator $U$ is smooth with respect to $D$, the $m$-th derivative of $k \mapsto e^{i k D} U e^{-i k D}$ is given by the commutator
$i^m e^{ik D} [D, [D, \cdots [D, U] \cdots]] e^{-ik D}$. We put the commutator $[\cdot, \cdot]$ $n$-times here.
In particular, $[D, [D, \cdots [D, U] \cdots]]$ is a bounded operator.
This is a consequence of the following lemma.

\begin{lemma}\label{lemma: smoothness}
Let $V \colon \Hil_1 \to \Hil_2$ be a bounded operator.
Let $D_1$ be a self-adjoint operator on $\Hil_1$ and let 
$D_2$ be a self-adjoint operator on $\Hil_2$.
Suppose that $\R \ni k \mapsto e^{ik D_2} V e^{-ik D_1} \in \mathcal{B}(\Hil_2 \leftarrow \Hil_1)$ is differentiable in the operator norm topology.
Then $V$ is a map from the domain of $D_1$ to that of $D_2$,
and
$D_2 V - V D_1 \colon {\rm dom}(D_1) \to \Hil_2$ is bounded with respect to the norm of $\Hil_1$.
The derivative of $k \mapsto e^{ik D_2} V e^{-ik D_1}$ is $i e^{ik D_2} (D_2 V - V D_1)e^{-ik D_1}$
\end{lemma}

\begin{proof}
Denote by $W$ the limit
\[W = \lim_{k \to 0} \dfrac{e^{ik D_2} V e^{-ik D_1} - V}{ik}\]
in the norm topology.
Let $\xi$ be an element of the domain of $D_1$.
Then we have
\begin{eqnarray*}
\frac{e^{ik D_2} - 1}{ik} V \xi 
=
\frac{e^{ik D_2} V e^{-ik D_1} - V}{ik} \xi
-
e^{ik D_2} V \frac{ e^{-ik D_1} - 1}{ik} \xi.
\end{eqnarray*}
As $k$ tends to $0$, the first term converges to $W \xi$.
The norm of $e^{ik D_2}$ is uniformly bounded and $e^{ik D_2}$ converges to $1$ in strong operator topology.
The vector $\frac{ e^{-ik D_1} - 1}{ik} \xi$ converges to $- D_1 \xi$ in norm.
Therefore, the vector $\frac{e^{ik D_2} - 1}{ik} V \xi$ converges to
$W \xi + V D_1 \xi$. It follows that
$V\xi \in {\rm dom}D_2$, $D_2 V \xi = W \xi + V D_1 \xi$.
We calculate the derivative as follows:
\begin{eqnarray*}
& &
\lim_{\Delta k \to 0} 
\frac{ 
e^{i (k + \Delta k) D_2} V e^{- i (k + \Delta k) D_1}
-
e^{i k D_2} V e^{- i k D_1}}
{\Delta k}
\\
&=&
i e^{i k D_2}
\lim_{\Delta k \to 0} 
\frac{ 
e^{i \Delta k D_2} V e^{- i \Delta k D_1} - V}
{i \Delta k}
e^{- i k D_1}
\\
&=&
i e^{i k D_2}
W
e^{- i k D_1}.
\end{eqnarray*}
\end{proof}

\begin{definition}\label{definition: regularity for QW}
A discrete-time or continuous-time quantum walk $\left( \mathcal{H}, \left( U^{(t)} \right), D \right)$ is said to be 
{\rm analytic} ({\rm smooth}, or {\rm uniform}), if for every $t$, $U^{(t)}$ is analytic (smooth, or uniform, respectively) with respect to $D$.
\end{definition}

\begin{remark}
In the case of continuous-time quantum walks, the author is not so sure about the above definition. A definition might be given by the relation between $D$ and the generator of the one-parameter unitary group $(U^{(t)})_{t \in \R}$, and would be stronger than our condition.
To state the main result of this paper,
we choose the weaker condition as above.
\end{remark}

\subsection{Basic examples of quantum walks}

The following are examples of quantum walks. The examples \ref{example: constant quantum walk} to \ref{example: the inverse Fourier transform} are analytic. We also define several notations, which are often used for the rest of this paper.

\begin{example}[Constant quantum walk]
\label{example: constant quantum walk}
Let $\alpha$ be a complex number whose absolute value is $1$.
The triplet $(\Hil, (\alpha^t)_{t \in \Z}, D)$ is a discrete-time quantum walk.
\end{example}

\begin{example}
[Discrete-time free quantum walk]
\label{example: discrete free QW}
Let $r$ be a positive real number.
Let $D_r$ be the diagonal operator on $\ell_2 (r \Z)$ defined by $D_r(\delta_x) = x \delta_x, x \in r \Z$.
Denote by $S_r$ the bilateral shift $\delta_x \mapsto \delta_{x+r}, x \in r \Z$.
Then $(\ell_2(r \Z), (S_r^t)_{t \in \Z}, D_r)$ is a one-dimensional discrete-time quantum walk. 
We call $(\ell_2(r \Z), (S_r^t)_{t \in \Z}, D_r)$ {\it the discrete-time free quantum walk}.
The map
\[\R \ni k \mapsto \exp(i k D_r) S_r \exp(- i k D_r) \in \B(\ell_2(r \Z))\] 
is given by
$\exp(i k D_r) S_r \exp(- i k D_r) = \exp(i k r) S_r$.
This can be extended to a holomorphic map defined on the complex plane $\C$.

The positive number $r$ is often defined by $1$.
\end{example}

\begin{example}
[Continuous-time free quantum walk]
Let $D$ be the multiplication operator on $L^2 (\R)$ given by the function $x \mapsto x$ on $\R$.
Let $X$ be the differential operator $- \frac{d}{i dx}$ on $L^2 (\R)$.
The one-parameter unitary group $(\exp(it X))_{t \in \R}$ generated by $X$ is the translation operator given by
\[[\exp(it X)(\xi)](x) = \xi(x - t), \quad \xi \in L^2(\R), x \in \R.\]
Then $\left( L^2(\R), (\exp(i t X))_{t \in \R}, D_1 \right)$ is a continuous-time quantum walk. 
We call $( L^2(\R)$, $(\exp(it X))_{t \in \R}$, $D_1 )$ {\it the continuous-time free quantum walk}.
\end{example}

\begin{example}
[$(2 \times 2)$-matrix]
Let
$
u_x =
\left(
\begin{array}{cc}
a_x & b_x\\
c_x & d_x
\end{array}
\right), x \in \Z
$
be a two-sided infinite sequence of complex unitary matrices.
Define a unitary operator $U$ on $\ell_2(\Z) \otimes \C^2$ as follows:
\begin{eqnarray*}
U (\delta_x \otimes \delta_1) &=& a_x \delta_{x - 1} \otimes \delta_1 + c_x \delta_{x + 1} \otimes \delta_2,\\
U (\delta_x \otimes \delta_2) &=& b_x \delta_{x - 1} \otimes \delta_1 + d_x \delta_{x + 1} \otimes \delta_2, \quad x \in \Z.
\end{eqnarray*}
Let $D$ be the diagonal operator on $\ell_2(\Z)$ defined in Example \ref{example: discrete free QW}. (In the present case, $r$ is $1$).
Then $(\ell_2(\Z) \otimes \C^2, (U^t)_{t \in \Z}, D \otimes {\rm id})$ is a discrete-time quantum walk.
\end{example}

\begin{example}
[Homogeneous $(2 \times 2)$-quantum walk]
\label{example: 2by2 homogeneous}
In the previous example, consider the case that
$u_x$ is a constant sequence
$u_x =
\left(
\begin{array}{cc}
a & b\\
c & d
\end{array}
\right)$.
Then the unitary operator $U$ is given by
$U = \left(
\begin{array}{cc}
a S_1^{-1} & b S_1^{-1}\\
c S_1 & d S_1
\end{array}
\right)$.
The triplet
$(\ell_2(\Z) \otimes \C^2, (U^t), D_1 \otimes {\rm id})$ is a homogeneous discrete-time quantum walk.
\end{example}

\begin{example}\label{example: the inverse Fourier transform}
Denote by $\T$ the set of all the complex numbers
whose absolute values are $1$.
The dual group of $r \Z$ is given by
$\T_{2 \pi r^{-1}} = \R /(2 \pi r^{-1} \Z)$ via the coupling
\[r \Z \times \T_{2 \pi r^{-1}} \ni (x, k + 2 \pi r^{-1} \Z) \mapsto \exp(i x k) \in \T.\]
We distinguish $\T_{2 \pi r^{-1}}$ from $\T$ in this paper.
The subscript $2 \pi r^{-1}$ is equal to the length of the torus $\T_{2 \pi r^{-1}}$.
We denote by $c_x$ the character on $\T_{2 \pi r^{-1}}$ defined by $x \in r \Z$.
Denote by $\mathcal{F}_r \colon L^2(\T_{2 \pi r^{-1}}) \to \ell_2(r \Z)$
the Fourier transform given by $c_x \mapsto \delta_x$.

The inverse Fourier transform $\widehat{D_r} = \F_r^{-1} D_r \F_r$ of $D_r$ in Example \ref{example: discrete free QW} is
\[ \left[ \widehat{D_r}(\xi) \right] \left( e^{ik} \right) 
= 
\dfrac{1}{i}\frac{d \xi}{dk} \left( e^{ik} \right), 
\quad \xi \in {\rm C}^\infty(\T_{2 \pi r^{-1}}), 
k + 2 \pi r^{-1} \Z \in \T_{2 \pi r^{-1}}.\]
We simply denote by $D_r$ the inverse Fourier transform $\widehat{D_r}$.
Here $D_r$ stands for the differential operator $\frac{d}{i d k}$.
The inverse Fourier transform $\F_r^{-1} S_r \F_r$ of the bilateral shift $S_r$ in Example \ref{example: discrete free QW} is the multiplication operator $M[c_r]$.

The inverse Fourier transform of the discrete-time free quantum walk in Example \ref{example: discrete free QW} is
$(L^2(\T_{2 \pi r^{-1}}), (M[c_r]^t)_{r \in \Z}, D_r)$.
This is also a quantum walk.
\end{example}

\begin{example}\label{example: the inverse Fourier transform of 2 by 2}
The inverse Fourier transform $\widehat{U} = (\F_1^{-1} \otimes {\rm id}) U (\F_1 \otimes {\rm id})$ of $U$ in Example \ref{example: 2by2 homogeneous} is
\[\widehat{U} = \left(
\begin{array}{cc}
a M[c_1]^{-1} & b M[c_1]^{-1}\\
c M[c_1] & d M[c_1]
\end{array}
\right).\]
Here $M[c_1]\colon L^2(\T_{2 \pi}) \to L^2(\T_{2 \pi})$ is the multiplication operator given by $c_1$.
The triplet 
$\left( L^2(\T_{2 \pi}) \otimes \C^2, \left( \widehat{U}^t \right)_{t \in \Z}, D_1 \otimes {\rm id} \right)$ is also a quantum walk.
\end{example}

\begin{example}
[Quantum walk by a multiplication operator]
\label{example: QW by multiplication operator}
Let $\lambda \colon \T_{2 \pi r^{-1}} \to \T$ be a Borel function.
Denote by $M[\lambda] \colon L^2(\T_{2 \pi r^{-1}}) \to L^2(\T_{2 \pi r^{-1}})$ the multiplication operator given by $\lambda$.
The triplet $\left( L^2(\T_{2 \pi r^{-1}}), (M[\lambda]^t)_{t \in \Z}, D_r = \frac{d}{i d k} \right)$ is a discrete-time quantum walk.
This type of quantum walks will often appear in this paper.
The walk is analytic (smooth, or uniform), if $\lambda$ analytic (smooth, or continuous, respectively).
\end{example}

\begin{example}[Direct sum]
Let $\left(\Hil_1, \left(U_1^{(t)}\right), D_1\right)$ and $\left(\Hil_2, \left(U_2^{(t)}\right), D_2\right)$ be two continuous-time or discrete-time quantum walks.
Then $(\Hil_1 \oplus \Hil_2$, $(U_1^{(t)} \oplus U_2^{(t)} )$, $D_1 \oplus D_2 )$ is also a quantum walk.
We call it {\it the direct sum quantum walk}.
\end{example}

\begin{example}[Amplification]
Let $m$ be a natural number.
Let $\left( \Hil, \left( U^{(t)} \right), D \right)$ be a quantum walk.
Then $\left(\Hil \otimes \C^n, \left(U^{(t)} \otimes {\rm id}\right), D \otimes {\rm id}\right)$ is also a quantum walk. We call it {\it the amplification quantum walk}.
\end{example}

Analyticity, smoothness, and uniformity are preserved under direct sum and amplification.

\subsection{Intertwiners between two quantum walks and their regularity}
\begin{definition}
Let $(\mathcal{H}_1, (U_1^{t})_{t \in \Z}, D_1)$ and $(\mathcal{H}_2, (U_2^{t})_{t \in \Z}, D_2)$
be two one-dimensional discrete-time quantum walks.
A bounded operator $X \colon \Hil_1 \to \Hil_2$ is called an {\rm intertwiner} between them, if it satisfies $X U_1 = U_2 X$.
\end{definition}

An intertwiner between $(\Hil_1, (U^{t})_{t \in \Z}, D)$ and itself is nothing other than an operator $X \in \mathcal{B}(\Hil)$ which commutes with $U$.

\begin{definition}
If
the mapping $\R \ni k \mapsto e^{i k D_2} X e^{-ik D_1} \in \mathcal{B}(\Hil_2 \leftarrow \Hil_1)$ is continuous (or smooth),
the intertwiner $X$ is said to be {\rm uniform} (or {\rm smooth}) with respect to $D_1$ and $D_2$.
If there exists a holomorphic extension of the map 
\[\R \ni k \mapsto e^{i k D_2} X e^{-i k D_1} \in \mathcal{B}(\mathcal{H})\]
defined on a domain of the form 
$\{\kappa \in \C | -\delta < {\rm Im}(\kappa) < \delta \}$, 
$X$ is said to be {\rm analytic}.
\end{definition}

The bounded operator $X \colon \Hil_1 \to \Hil_2$ defines an operator
\[\widetilde{X} = 
\left(
\begin{array}{cc}
0 & 0\\
X & 0\\
\end{array}
\right)
\colon \Hil_1 \oplus \Hil_2 
\to 
\Hil_1 \oplus \Hil_2 .
\]
The operator $X$ is an intertwiner between $U_1$ and $U_2$, if and only if $\widetilde{X}$ commutes with $U_1 \oplus U_2$.
The intertwiner $X$ is uniform (smooth, or analytic) with respect to $D_1$ and $D_2$,
if and only if $\widetilde X$ is uniform (smooth, or analytic, respectively) with respect to $D_1 \oplus D_2$.

Let $(\mathcal{H}_3, (U_3^{t})_{t \in \Z}, D_3)$ be another quantum walk.
If $X_1 \colon \Hil_1 \to \Hil_2$ is an intertwiner between $U_1$ and $U_2$, and if
$X_2 \colon \Hil_2 \to \Hil_3$ is an intertwiner between $U_2$ and $U_3$,
then
If $X_2 X_1 \colon \Hil_1 \to \Hil_3$ is an intertwiner between $U_1$ and $U_3$.
If $X_1$ is uniform with respect to $D_1$ and $D_2$, and if $X_2$ is uniform with respect to $D_2$ and $D_3$, then
$X_2 X_1$ is uniform with respect to $D_1$ and $D_3$.
Smoothness and analyticity are also preserved under this composition procedure.

To determine whether there exists a non-zero uniform intertwiner between given two homogeneous quantum walks, we use the following as a key tool.

\begin{proposition}\label{proposition: converging to an atom}
Let $r(1)$ and $r(2)$ be positive real numbers.
Let $V$ be a bounded operator from $\ell_2(r(1) \Z)$ to $\ell_2(r(2) \Z)$.
Let $D_{r(1)}$ be the diagonal operator $\delta_y \mapsto y \delta_y, y \in r(1) \Z$ on $\ell_2 (r(1) \Z)$. 
Let $D_{r(2)}$ be the diagonal operator $\delta_x \mapsto x \delta_x, x \in r(2) \Z$ on $\ell_2 (r(2) \Z)$. 
Suppose that
$V$ is uniform with respect to $D_{r(1)}$ and $D_{r(2)}$.
For $y \in r(1) \mathbb{N}$, define a probability measure $p_y$ on $\mathbb{R}$ by
\[ p_y = \sum_{x \in r(2) \Z} \left| \left\langle V \delta_y, \delta_x \right\rangle_{\ell_2(r(2) \Z)} \right|^2 \delta_{x / y},\]
where $\delta_{y / x}$ stands for the point mass at $y / x \in \R$.
Then for every positive number $\delta$, we have
\[\lim_{y \to \infty} p_y((-\infty, 1 - \delta] \cup [1 + \delta, \infty)) = 0\]
\end{proposition}

\begin{proof}
For every real number $k$, the matrix coefficient of $\exp (i k D_{r(2)}) V \exp (-ik D_{r(1)})$ at $(x, y)$ is
\[\left\langle \exp (i k D_{r(2)}) V \exp (-ik D_{r(1)}) \delta_y, \delta_x \right\rangle_{\ell_2(r(2) \Z)}
= \exp(ik (x - y)) \left\langle V \delta_y, \delta_x \right\rangle_{\ell_2(r(2) \Z)}.
\]
For a positive real number $\sigma$, define an operator $V_\sigma$ by
\[V_\sigma = \int_{-\infty}^\infty \exp \left( -\frac{k^2}{2 \sigma^2} \right) \exp (i k D_{r(2)}) V \exp (-ik D_{r(1)}) \dfrac{dk}{\sqrt{2 \pi} \sigma}.\]
Note that the operator norm of $V_\sigma$ is no more than that of $V$.
We also note that as $\sigma$ tends to $0$, $V_\sigma$ converges to $V$ in the operator norm topology.
The matrix coefficient of $V_\sigma$ is
\begin{eqnarray*}
\left\langle V_\sigma \delta_y, \delta_x \right\rangle_{\ell_2(r(2) \Z)}
&=& \int_{-\infty}^\infty \exp \left( -\frac{k^2}{2 \sigma^2} \right) \exp(ik (x - y)) \dfrac{dk}{\sqrt{2 \pi} \sigma} \left\langle V \delta_y, \delta_x \right\rangle_{\ell_2(r(2) \Z)}\\
&=& 
\exp\left( -\frac{(x - y)^2}{2} \sigma^2 \right)
\left\langle V \delta_y, \delta_x \right\rangle_{\ell_2(r(2) \Z)}.
\end{eqnarray*}

Take arbitrary (small) positive real numbers $\delta$ and $\epsilon$.
There exists a (small) positive real number $\sigma$ such that
$\| V - V_\sigma \| < \epsilon$.
For such $\epsilon$ and $\sigma$,
there exists a positive number $K$ such that for every $y \in r(1) \Z$, 
\[\sum_{x \in r(2) \Z, |x - y| \ge K} \exp\left( - {(x - y)^2} \sigma^2 \right) < \epsilon.\]
We consider the case that $y \in r(1) \Z$ is larger than $K / \delta$.
Then
\begin{eqnarray*}
p_y(\R \setminus (1 - \delta, 1 + \delta))
&=&
\sum_{x \in r(2) \Z, |x / y - 1| \ge \delta} \left| \left\langle V \delta_y, \delta_x \right\rangle_{\ell_2(r(2) \Z)} \right|^2.
\end{eqnarray*}
Since the inequality $|x / y - 1| \ge \delta$ implies $|x - y| \ge |y| \delta \ge K$,
we have
\begin{eqnarray*}
p_y(\R \setminus (1 - \delta, 1 + \delta))
&\le&
\sum_{x \in r(2) \Z, |x - y| \ge K} \left| \left\langle V \delta_y, \delta_x \right\rangle_{\ell_2(r(2) \Z)} \right|^2.
\end{eqnarray*}
We further obtain the following inequalities
\begin{eqnarray*}
& &p_y(\R \setminus (1 - \delta, 1 + \delta))\\
&\le&
\sum_{x \in r(2) \Z, |x - y| \ge K} 
\left| \left\langle V_\sigma \delta_y, \delta_x \right\rangle 
+ \left\langle (V - V_\sigma) \delta_y, \delta_x \right\rangle \right|^2\\
&\le&
2 \sum_{x \in r(2) \Z, |x - y| \ge K}
\left| \left\langle V_\sigma \delta_y, \delta_x \right\rangle  \right|^2
+ 
2 \sum_{x \in r(2) \Z, |x - y| \ge K}
\left| \left\langle (V - V_\sigma) \delta_y, \delta_x \right\rangle \right|^2\\
&\le&
2 \sum_{x \in r(2) \Z, |x - y| \ge K} 
\exp\left( - {(x - y)^2} \sigma^2 \right) \left| \left\langle V \delta_y, \delta_x \right\rangle  \right|^2
+ 
2 \left\| (V - V_\sigma) \delta_y \right\|^2.
\end{eqnarray*}
By the assumptions on $\sigma$ and $K$, we have
\begin{eqnarray*}
p_y(\R \setminus (1 - \delta, 1 + \delta))
\le
2 \|V\| \epsilon
+ 
2 \|V - V_\sigma\|^2
\le
2 \|V\| \epsilon
+ 
2 \epsilon^2.
\end{eqnarray*}
It follows that the positive measure $p_y$ tends to $0$ on $\R \setminus (1 - \delta, 1 + \delta)$.
\end{proof}

\begin{remark}\label{remark: uniform Roe algebra}
A bounded operator $V \colon \ell_2 \Z \to \ell_2 \Z$ is continuous with respect to the standard diagonal operator $D_1 \colon \delta_x \to x \delta_x, x \in \Z$, if and only if $V$ is an element of the uniform Roe algebra ${\rm C}^*_{\rm u} (\Z)$ defined in \cite[Subsection 4.5]{Roe}.
We can easily prove it, using $V_\sigma$ introduced in the above proof.
In this paper, we regard ${\rm C}^*_{\rm u} (\Z)$ as the space of operators on $\ell_2(\Z)$ which are uniform with respect to $D_1$.
\end{remark}

\subsection{Similarity between discrete-time quantum walks}

\begin{definition}
\label{definition: similarity}
Discrete-time quantum walks $(\mathcal{H}_1, (U_1^{t})_{t \in \Z}, D_1)$ and $(\mathcal{H}_2$, $(U_2^{t})_{t \in \Z}$, $D_2)$ are said to be {\rm similar}, if
there exists a unitary operator $V \colon \Hil_1 \to \Hil_2$ which is a {\it smooth} intertwiner between $(\mathcal{H}_1, (U_1^{t})_{t \in \Z}, D_1)$ and $(\mathcal{H}_2, (U_2^{t})_{t \in \Z}, D_2)$.
If $V$ maps the domain of $D_1$ to that of $D_2$ and $D_2 V = V D_1$ holds, or equivalently, if the mapping $k \mapsto \exp(i k D_2) V \exp(-ik D_1)$ is constant, then these walks are said to be {\rm unitary equivalent}.
\end{definition}

If $(\mathcal{H}_1, (U_1^{t})_{t \in \Z}, D_1)$ 
and $(\mathcal{H}_2, (U_2^{t})_{t \in \Z}, D_2)$ are similar, 
and if one of them is smooth, then 
the other is also smooth.
Similarity is an equivalence relation on smooth quantum walks.

\begin{example}
Example \ref{example: discrete free QW} and Example
\ref{example: the inverse Fourier transform} are unitary equivalent.
Example \ref{example: 2by2 homogeneous} and Example \ref{example: the inverse Fourier transform of 2 by 2} are also unitary equivalent.
The Fourier transform is a smooth intertwining operator.
\end{example}

\begin{example}\label{example: replacing D}
Let $(\Hil, (U^{t})_{t \in \Z}, D_1)$ be a smooth quantum walk.
Let $D_2$ be a self-adjoint operator on $\Hil$.
If the mapping $\R \ni k \mapsto e^{ik D_2} e^{- i k D_1}$ is smooth, then the
quantum walks $(\Hil, (U^{t})_{t \in \Z}, D_1)$ and $(\Hil, (U^{t})_{t \in \Z}, D_2)$ are similar.
Indeed, the identity operator gives a smooth intertwining operator between them.
\end{example}

\begin{example}
Let $(\Hil, (U^{t})_{t \in \Z}, D)$ be a smooth quantum walk.
Let $V$ be a unitary operator on $\Hil$. If $V$ is smooth with respect to $D$, then
quantum walks $(\Hil, (U^{t})_{t \in \Z}, D)$ and $(\Hil, (V U^{t} V^{-1})_{t \in \Z}, D)$ are similar.
Indeed, the unitary operator $V$ gives a smooth intertwining operator between them.
\end{example}

Similarity is compatible with direct sum and with amplification.

\subsection{Asymptotic behavior of quantum walks}

Let $(\Hil, (U^t)_{t \in \Z}, D)$ be a discrete-time smooth quantum walk.
Fix a unit vector $\xi$ in $\Hil$.
We often call $\xi$ an initial unit vector of the quantum walk.
Let $E(\cdot)$ be the spectral projection of $D$.
Recall that $E$ maps Borel subsets of $\R$ to orthogonal projection in $\mathcal{B}(\Hil)$.
For every $t \in \mathbb{N}$, we have a probability measure on $\R$ defined by
$\langle E(\cdot) U^t \xi, U^t \xi \rangle$. We pay attention on the push-forward $p_t$ of the measure under the mapping $\R \ni x \mapsto x / t \in \R$.
The measure $p_t$ is given by
\[p_t(\Omega) = \left\langle E(t \Omega) U^t \xi, U^t \xi \right\rangle, \quad {\rm a\  Borel\ subset\ } \Omega \subset \R. \]

The vector $\xi \in \Hil$ is said to be {\it smooth} with respect to $D$, if $\xi \in \bigcap_{m \in \N} \mathrm{dom} (D^m)$.
We often assume that 
the quantum walk $(\Hil, (U^t)_{t \in \Z}, D)$ is smooth and that the initial unit vector
$\xi$ is smooth.
It is not hard to see that for every integer $t$, $U^t \xi$ is also smooth with respect to $D$.

\begin{lemma}\label{lemma: moment}
The $m$-th moment of $p_t$ is
$\displaystyle \int_{v \in \R} v^m \cdot p_t(dv) 
=
\left\langle \dfrac{1}{t^m} D^m  U^t \xi, U^t \xi \right\rangle.$
\end{lemma}

\begin{proof}
For every $t \in \N$, 
we calculate the $m$-th moment of $p_t$ as follows:
\begin{eqnarray*}
\int_{v \in \R} v^m \cdot p_t(dv) 
&=&
\int_{v \in \R} v^m \cdot \langle E(t dv) U^t \xi, U^t \xi \rangle \\
&=&
\int_{x \in \R} \dfrac{1}{t^m} x^m \cdot \langle E(dx) U^t \xi, U^t \xi \rangle.
\end{eqnarray*}
This is nothing other than the right hand side of the lemma.
\end{proof}

\begin{definition}
If the weak limit of $p_t$ exists, it is called the {\rm limit distribution} of the quantum walk $(\Hil, (U^t)_{t \in \Z}, D)$ with respect to the vector $\xi$.
\end{definition}

\begin{lemma}\label{lemma: upper bound for moments}
Let $(\Hil, (U^t)_{t \in \Z}, D)$ be a discrete-time {\rm smooth} quantum walk.
Let $\xi \in \Hil$ be a unit vector.
Assume that $\xi$ is smooth with respect to $D$.
Then we have
\begin{eqnarray*}
\limsup_t \left| \int_{v \in \R} v^m \cdot p_t(dv) \right| 
\le \|[D, U]\|^m.
\end{eqnarray*}
\end{lemma}

\begin{proof}

For a while, we fix $m$ and consider the case that $t$ is large.
By Lemma \ref{lemma: smoothness}, we can define a sequence of bounded operators $u_0, u_1, u_2, \cdots$ as follows:
\[u_0 = U, \quad u_1 = [D, u_0], \quad u_2 = [D, u_1], \cdots.\]
By smoothness of $\xi$, we can also define a sequence of vectors $\xi_0, \xi_1, \xi_2, \cdots \in \Hil$ by
\[\xi_t = D^t \xi, \quad t \in \N.\]
For smooth operators and smooth vectors, the following Leibniz rule holds:
\begin{eqnarray*}
D X \eta &=& X' \eta + X \eta', \quad {\rm where }\ X' = [D, X], \eta' = D \eta,
\\
\left[D, X Y\right] &=& X' Y + X Y', \quad {\rm where } \ X' = [D, X], Y' = [D, Y].
\end{eqnarray*}
By the Leibniz rule, the vector $D^m  U^t \xi$ can be expressed as follows:
\begin{eqnarray}\label{equation: Leibniz rule}
D^m  U^t \xi = \sum_{s \in I} u_{\sharp s^{-1}(t)} u_{\sharp s^{-1}(t - 1)} \cdots u_{\sharp s^{-1}(1)} \xi_{\sharp s^{-1}(0)}.
\end{eqnarray}
In this formula,
\begin{itemize}
\item
$s$ is an element of the index set
\[ I = \left\{ s \colon \{1, 2, \cdots, m\} \to \{0, 1, \cdots, t\} \ | \ \rm{a\ map}\right\},\]
\item
$s^{-1} (j)$ is the inverse image of $\{j\} \subset  \{0, 1, \cdots, t\}$ under the mapping $s$.
\item
$\sharp s^{-1} (j)$ is the number of elements of the inverse image.
\end{itemize}
Define a subset $I_0$ of $I$ as follows:
\[ I_0 = \left\{ s \colon \{1, 2, \cdots, m\} \to \{0, 1, \cdots, t\} \in I \ | \ 
s^{-1}(0) = \emptyset, s {\rm\ is\ injective}
\right\}.\]
It is not hard to see that if $t$ is large, almost all the elements of $I$ are in $I_0$.
More precisely, $\lim_{t \to \infty} \sharp I_0 / \sharp I$ is $1$.
For $s \in I$, the norm of the term $u_{\sharp s^{-1}(t)}$ $u_{\sharp s^{-1}(t - 1)}$ $\cdots$ $u_{\sharp s^{-1}(1)}$ $\xi_{\sharp s^{-1}(0)}$ in the equation (\ref{equation: Leibniz rule})
is bounded by 
\[\max\{1 = \|u_0\|, \|u_1\|, \cdots, \|u_m\| \}^m \max\{1 = \|\xi\|, \|\xi_1\|, \cdots, \|\xi_m\|\}.\]
It follows that
\begin{eqnarray*}
\limsup_t \left\| \dfrac{1}{t^m} D^m  U^t \xi \right\|
&\le&
\limsup_t 
\dfrac{1}{\sharp I}
\sum_{s \in I} \left\| u_{\sharp s^{-1}(t)} u_{\sharp s^{-1}(t - 1)} \cdots u_{\sharp s^{-1}(1)} \xi_{\sharp s^{-1}(0)} \right\|\\
&=&
\limsup_t 
\dfrac{1}{\sharp I}
\sum_{s \in I_0} \left\| u_{\sharp s^{-1}(t)} u_{\sharp s^{-1}(t - 1)} \cdots u_{\sharp s^{-1}(1)} \xi_{\sharp s^{-1}(0)} \right\|\\
&\le&
\limsup_t 
\dfrac{\sharp I_0}{\sharp I}
\|u_1\|^m
\\
&=&
\|[D, U]\|^m
\end{eqnarray*}
Combining with the equation in Lemma \ref{lemma: moment}, for every positive integer $m$, we have
\begin{eqnarray*}
\limsup_t \left| \int_{v \in \R} v^m \cdot p_t(dv) \right| 
=
\limsup_t \left| \left\langle \dfrac{1}{t^m} D^m  U^t \xi, U^t \xi \right\rangle \right|
\le
\|[D, U]\|^m.
\end{eqnarray*}
\end{proof}

\begin{proposition}\label{proposition: compact support}
Let $(\Hil, (U^t)_{t \in \Z}, D)$ be a discrete-time smooth quantum walk.
Let $\xi \in \Hil$ be a smooth unit vector.
For every $L \in \R$ larger than $\|[D, U]\|$,
we have
\[\lim_{t \to \infty} p_t((-\infty, -L] \cup [L, \infty)) = 0.\]
\end{proposition}

\begin{proof}
By Lemma \ref{lemma: upper bound for moments},
for every positive integer $m$, we obtain the following:
\begin{eqnarray*}
\limsup_t p_t ((-\infty, -L] \cup [L, \infty))
&\le&
\limsup_t \dfrac{1}{L^{2m}} \int_{v \in (-\infty, -L] \cup [L, \infty)}  v^{2m} \cdot p_t (dv)\\
&\le&
\limsup_t \dfrac{1}{L^{2m}} \int_{v \in \R}  v^{2m} \cdot p_t (d v)\\
&\le&
\dfrac{ \|[D, U]\|^{2m}}{L^{2m}}.
\end{eqnarray*}
Since the positive integer $m$ is arbitrary, the conclusion follows.
\end{proof}

\begin{corollary}
Let $(\Hil, (U^t)_{t \in \Z}, D)$ be a discrete-time smooth quantum walk.
Let $\xi \in \Hil$ be a smooth unit vector.
If the limit distribution exists, then its support is compact.
\end{corollary}

\begin{proposition}\label{proposition: two kinds of convergence}
Let $(\Hil, (U^t)_{t \in \Z}, D)$ be a discrete-time smooth quantum walk.
Let $\xi \in \Hil$ be a smooth unit vector.
Let $p_\infty$ be a Borel measure on $\R$.
The sequence of the measures $p_t$ weakly converges to $p_\infty$, if and only if it converges to $p_\infty$ in law. 
\end{proposition}

`If part' of this proposition is a consequence of the general theory like \cite[Theorem 4.5.5]{Chung}.
We give a proof to make the argument self-contained.

\begin{proof}
If $p_t$ converges to $p_\infty$ in law,
then the support of $p_\infty$ is included in $[- \|[D, U]\|$, $\|[D, U]\|]$, 
by Lemma \ref{lemma: upper bound for moments}.
If $p_t$ weakly converges to $p_\infty$, 
then the support of $p_\infty$ is included in $[- \|[D, U]\|$, $\|[D, U]\|]$, 
by Proposition \ref{proposition: compact support}.

Let $\epsilon$ be an arbitrary positive real number less than $1$.
Take a real number $L$ larger than $\|[D, U]\|$.
For a bounded continuous function $f$ on $\R$,
and for a polynomial function $g$ on $\R$ satisfying
\begin{eqnarray}\label{equation: f and g}
|g(v) - f(v)| < \epsilon, \quad v \in [- L, L],
\end{eqnarray}
we have
\begin{eqnarray}\label{equation: p infty}
\left| \int_{v \in \R} f(v) p_\infty(dv) - \int_{v \in \R} g(v) p_\infty(dv) \right| < \epsilon.
\end{eqnarray}
Since the function $f$ is bounded and $g$ is polynomial, there exists a positive integer $m$
such that
\begin{itemize}
\item
for $v \in ( - \infty, - L] \cup [L, \infty)$, 
$| f(v) - g(v) | < \left( \dfrac{v}{L} \right)^{2m}$, and
\item
$\left( \dfrac{\|[D, U]\|}{L} \right)^{2m} < \epsilon$.
\end{itemize}
For such a natural number $m$, we have
\begin{eqnarray*}
& &
\left| \int_{v \in \R} f(v) p_t(dv) - \int_{v \in \R} g(v) p_t(dv) \right|\\
&\le&  
\int_{-L}^L | f(v) - g(v) | p_t(dv) + \int_{v \in ( - \infty, - L] \cup [L, \infty)} | f(v) - g(v) | p_t(dv) \\
&\le&  
\epsilon + \int_{v \in \R} \left( \dfrac{v}{L} \right)^{2m} p_t(dv).
\end{eqnarray*}
By Lemma \ref{lemma: upper bound for moments},
if $t$ is large enough,
then we have
\begin{eqnarray}\label{equation: p t}
\left| \int_{v \in \R} f(v) p_t(dv) - \int_{v \in \R} g(v) p_t(dv) \right|
\le
\epsilon + \left( \dfrac{\|[D, U]\|}{L} \right)^{2m} + \epsilon < 3 \epsilon.
\end{eqnarray}

Suppose
that $p_t$ converges to $p_\infty$ in law.
Take an arbitrary bounded continuous function $f$ on $\R$.
There exists a polynomial function $g$ satisfying the inequality
(\ref{equation: f and g}), by the approximation theorem of Weierstrass.
The inequality (\ref{equation: p infty}) follows.
If $t$ is large enough,
we obtain the inequality (\ref{equation: p t}).
By the definition of convergence in law, if $t$ is large enough,  
\begin{eqnarray}\label{equation: g}
\left| \int_{v \in \R} g(v) p_t(dv) - \int_{v \in \R} g(v) p_\infty(dv) \right| < \epsilon.
\end{eqnarray}
The inequalities (\ref{equation: p infty}), (\ref{equation: p t}), and (\ref{equation: g})
implies 
that if $t$ is large enough,
\begin{eqnarray*}
\left| \int_{v \in \R} f(v) p_t(dv) - \int_{v \in \R} f(v) p_\infty(dv) \right|
< 5 \epsilon.
\end{eqnarray*}
We conclude that $p_t$ weakly converges to $p_\infty$.

Conversely suppose that $p_t$ weakly converges to $p_\infty$.
Take an arbitrary polynomial function $g$.
Then there exists a bounded continuous function $f$ on $\R$ satisfying the inequality
(\ref{equation: f and g}).
Then we have the inequality (\ref{equation: p infty}).
If $t$ is large enough,
we obtain the inequality (\ref{equation: p t}).
By the definition of weak convergence, if $t$ is large enough,  
\begin{eqnarray}\label{equation: f}
\left| \int_{v \in \R} f(v) p_t(dv) - \int_{v \in \R} f(v) p_\infty(dv) \right| < \epsilon.
\end{eqnarray}
The inequalities (\ref{equation: p infty}), (\ref{equation: p t}), and (\ref{equation: f})
implies 
that if $t$ is large enough,
\begin{eqnarray*}
\left| \int_{v \in \R} g(v) p_t(dv) - \int_{v \in \R} g(v) p_\infty(dv) \right|
< 5 \epsilon.
\end{eqnarray*}
We conclude that $p_t$ converges to $p_\infty$ in law.
\end{proof}

When we discuss the limit distribution,
we can freely replace the original quantum walk with similar one.

\begin{theorem}\label{theorem: similarity and limit distribution}
Assume that two smooth quantum walks $(\mathcal{H}_1, (U_1^{t})_{t \in \Z}, D_1)$ and $(\mathcal{H}_2, (U_2^{t})_{t \in \Z}, D_2)$ are similar.
Let $V \colon \Hil_1 \to \Hil_2$ be a unitary operator which gives similarity between the quantum walks. Let $\xi$ be a unit vector in $\Hil_1$ which is smooth with respect to $D_1$.
Then the vector $V \xi$ is smooth with respect to $D_2$. The quantum walk $(\mathcal{H}_2, (U_2^{t})_{t \in \Z}, D_2)$ has limit distribution with respect to $V \xi$, if and only if  $(\mathcal{H}_1, (U_1^{t})_{t \in \Z}, D_1)$ has limit distribution with respect to $\xi$. In this case, these limit distributions coincide.
\end{theorem}

\begin{proof}
By Lemma \ref{lemma: smoothness}, $V$ maps $\xi$ in $\mathrm{dom} (D_1^m)$ to an element of $\mathrm{dom} (D_2^m)$.

Let $p_{1, t}$ be the $t$-th probability measure of $(\mathcal{H}_1, (U_1^{t})_{t \in \Z}, D_1)$ with respect to $\xi$.
Let $p_{2, t}$ be the $t$-th probability measure of $(\mathcal{H}_2, (U_2^{t})_{t \in \Z}, D_2)$ with respect to $V \xi$.
By Lemma \ref{lemma: moment}, the $m$-th moment of $p_{1, t}$ is given by
\[\left\langle \dfrac{1}{t^m} D_1^m U_1^t  \xi, U_1^t \xi \right\rangle.\]
The $m$-th moment of $p_{2, t}$ is given by
\[
\left\langle \dfrac{1}{t^m} D_1^m U_2^t V \xi, U_2^t V \xi \right\rangle
=
\left\langle \dfrac{1}{t^m} D_2^m V U_1^t \xi, V U_1^t \xi \right\rangle.
\]
By Proposition \ref{proposition: two kinds of convergence}, it suffices to show that for every $m$, as $t$ tends to infinity, the difference of these moments converges to $0$.

Define sequences of bounded operators $\{v_j \colon \Hil_1 \to \Hil_2\}$ and $\{u_j \colon \Hil_1 \to \Hil_1\}$ by
\begin{eqnarray*}
v_0 = V, & & v_j = D_2 v_{j -1} - v_{j -1} D_1,\\
u_0 = U_1, & & u_j = D_1 u_{j -1} - u_{j -1} D_1.
\end{eqnarray*}
By Lemma \ref{lemma: smoothness}, these operators are bounded.
Define vectors $\{\xi_j\} \in \Hil_1$ by
\begin{eqnarray*}
\xi_0 = \xi, \quad \xi_j = D_1 \xi_{j -1}.
\end{eqnarray*}
By the Leibniz rule, the vector $D_2^m V U_1^t \xi$ can be expressed as follows:
\[D_2^m V U_1^t \xi = \sum_{s \in J} v_{\sharp s^{-1}(t + 1)} u_{\sharp s^{-1}(t)} u_{\sharp s^{-1}(t - 1)} \cdots u_{\sharp s^{-1}(1)} \xi_{\sharp s^{-1}(0)}.   \]
In the formula,
\begin{itemize}
\item
$s$ is an element of the index set
\[ J = \left\{ s \colon \{1, 2, \cdots, m\} \to \{0, 1, \cdots, t, t + 1\} \ | \ \rm{a\ map}\right\},\]
\item
$s^{-1} (j)$ is the inverse image of $\{j\} \subset  \{0, 1, \cdots, t, t + 1\}$ under the mapping $s$.
\item
$\sharp s^{-1} (j)$ is the number of elements of the inverse image.
\end{itemize}
Define a subset $J_0$ of $J$ as follows:
\[ J_0 = \left\{ s \colon \{1, 2, \cdots, m\} \to \{0, 1, \cdots, t, t + 1\} \in J \ | \ 
s^{-1}(t+ 1) = \emptyset \right\}.\]
If $t$ is large, the number of elements in the coset $J \setminus J_0$ is much smaller than $t^m$.
That is $\lim_{t \to \infty} (\sharp J - \sharp J _0)/ t^m = 1$.
For $s \in J$, the norm of the term 
\[v_{\sharp s^{-1}(t + 1)} u_{\sharp s^{-1}(t)} u_{\sharp s^{-1}(t - 1)} \cdots u_{\sharp s^{-1}(1)} \xi_{\sharp s^{-1}(0)}\] in $D_2^m V U_1^t \xi$
 is bounded by 
\[\max_{0 \le j \le m} \|v_j\| \left( \max_{0 \le j \le m} \|u_j\|\right)^m \max_{0 \le j \le m} \|\xi_j\|.\]
It follows that
\begin{eqnarray*}
\lim_{t \to \infty} 
\left\| \dfrac{1}{t^m} D_2^m V U_1^t \xi 
-
\dfrac{1}{t^m} \sum_{s \in J_0} v_{\sharp s^{-1}(t + 1)} u_{\sharp s^{-1}(t)} u_{\sharp s^{-1}(t - 1)} \cdots u_{\sharp s^{-1}(1)} \xi_{\sharp s^{-1}(0)}
\right\|
=
0.
\end{eqnarray*}
The second term in the limit is nothing other than the following vector:
\begin{eqnarray*}
\dfrac{1}{t^m} V \sum_{s \in J_0} u_{\sharp s^{-1}(t)} u_{\sharp s^{-1}(t - 1)} \cdots u_{\sharp s^{-1}(1)} \xi_{\sharp s^{-1}(0)}
=
\dfrac{1}{t^m} V D_1^m U_1^t \xi.
\end{eqnarray*}
It follows that the $m$-th moment
$
\left\langle \dfrac{1}{t^m} D_2^m V U_1^t \xi, V U_1^t \xi \right\rangle
$
of $p_{2, t}$
is asymptotically identical to
\[
\left\langle \dfrac{1}{t^m} V D_1^m U_1^t \xi, V U_1^t \xi \right\rangle
=
\left\langle \dfrac{1}{t^m} D_1^m U_1^t \xi, U_1^t \xi \right\rangle.
\]
This is nothing other than the $m$-th moment of $p_{1, t}$.
\end{proof}

\subsection{Homogeneous quantum walks}
\label{subsection: homogeneous}
We propose the following axiom on one-dimensional homogeneous quantum walks.

\begin{definition}\label{definition: homogeneous}
The quadruple $U = (\Hil, (U^t)_{t \in \Z}, D, S)$ is called a discrete-time {\rm homogeneous} quantum walk, if the following conditions hold:
\begin{itemize}
\item
The triple $(\Hil, (U^t)_t, D)$ is a quantum walk.
\item
$S$ is a unitary operator on $\Hil$.
\item
$U S = S U$.
\item
$S$ preserves the domain of $D$.
\item
$S^{-1} D S - D$ is a positive constant operator $r \cdot {\rm id}$.
\item
The spectral projection of $D$ corresponding to $[0, r) \subset \R$ has finite rank.
\end{itemize}
An operator $X$ on $\Hil$ is said to be homogeneous, if $X$ commutes with $S$.
It is said to be essentially homogeneous, if there exists a natural number $N$ such that the operator $X$ commutes with $S^N$.
Regularity (uniformity, smoothness, or analyticity) for an operator on $\Hil$ is determined by $D$ as in Definition $\ref{definition: regularity}$.

Two discrete-time homogeneous quantum walks are said to be similar, if
two triplets of quantum walks are similar, and if the forth entries are unitary equivalent via the intertwiner which gives similarity.
If we need to consider continuous-time homogeneous quantum walk, replace $(U^t)_{t \in \Z}$ with $(U^{(t)})_{t \in \R}$.
\end{definition}

Let $\Hil_0$ denote the spectral subspace of $D$ corresponding to $[0, r) \subset \R$.
Denote by $n$ the dimension of $\Hil_0$.
The natural number $n$ is called {\it the degree of freedom}.
The Hilbert space $\Hil$ is decomposed as follows:
\[\Hil = \overline{\cdots \oplus S^{-1} \Hil_0 \oplus \Hil_0 \oplus S \Hil_0 \oplus \cdots}.\]
Identifying $\Hil_0$ with $\C^n$,
we can easily show that $(\Hil, (U^t)_t, D)$ is similar to a quantum walk $(\ell_2(\Z) \otimes \C^n, (U^t)_t, D_1 \otimes {\rm id})$ and that $S$ is identified with the bilateral shift $S_1 \otimes {\rm id}$.
We may fix the original homogeneous quantum walk $U$ as $U = (\ell_2(\Z) \otimes \C^n, (U^t)_t, D_1 \otimes {\rm id}, S_1 \otimes {\rm id})$.
For the rest of this paper, we always assume that $U$ is of the form
\[(\ell_2(\Z) \otimes \C^n, (U^t)_{t \in \Z}, D_1 \otimes {\rm id}, S_1 \otimes {\rm id})\]
and that the generator $U$ is analytic with respect to $D_1 \otimes {\rm id}$.

\section{Structure theorems on intertwiners and commutant}

In this section, 
we demonstrate the way to determine the space of {\it uniform} intertwining operators between discrete-time homogeneous {\it analytic} quantum walks.
As a corollary, the algebra of uniform operators which commute with a given walk is determined.
For the first half, we review the structure theorem in \cite{SaigoSako}.
We need to fix the notations related to Fourier analysis.

\subsection{Fourier analysis}\label{subsection: Fourier analysis}
We consider the group $r \Z$ generated by a positive real number $r$ and its dual.
All the characters of $r \Z$ are of the form
\[\widehat \chi_k (x) = \exp(i k x), \quad k \in \left[ 0, \dfrac{2 \pi}{r} \right), x \in r \Z. \]
We identify the dual group $\{\widehat{\chi}_k \ |\ 0 \le k < 2 \pi r^{-1} \}$ with
$\R / (2 \pi r^{-1} \Z)$.
We often denote by $\T_{2 \pi r^{-1}}$ the dual group $\R / (2 \pi r^{-1} \Z)$.
The subscript $2 \pi r^{-1}$ is equal to the length of the torus.
We introduce the counting measure on $r \Z$. The scalar multiple $\dfrac{r}{2 \pi} dk$ of the Lebesgue measure $dk$ defines the Haar measure on $\T_{2 \pi r^{-1}}$.
For $x \in r \Z$, the character $c_x$ on $\T_{2 \pi r^{-1}}$ is defined as
\[c_x(k) = \exp(i k x), \quad k \in \T_{2 \pi r^{-1}}.\]
The Fourier transform $\F_r \colon L^2(\T_{2 \pi r^{-1}}) \to \ell_2(r \Z)$ maps
$c_x$ to the definition function $\delta_x$ of $\{x\} \subset r \Z$.

\subsection{Model quantum walk}

Let $\lambda \colon \T_{2 \pi r^{-1}} \to \T = \{z \in \C \ |\  |z| = 1\}$ be an analytic function.
The function $\lambda$ defines the multiplication operator $M[\lambda] \colon 
 L^2(\T_{2 \pi r^{-1}}) \to L^2(\T_{2 \pi r^{-1}})$. 
The triplet $(\ell_2(r \Z), (U_\lambda^{t})_{t \in \Z}, D_r)$
of 
\begin{itemize}
\item
The Hilbert space $\ell_2(r \Z)$ of the square summable functions on $r \Z$,
\item
The Fourier transform $U_\lambda^t = \F_r M[\lambda]^t \F_r^{-1} \colon  \ell_2(r \Z) \to  \ell_2(r \Z)$,
\item
The diagonal operator $D_r$ given by $D_r(\delta_x) = x \delta_x, x \in r \Z$.
\end{itemize}
is called a {\it model quantum walk}.
Here we note that the inverse Fourier transform $(\F_r^{-1} \ell_2(r \Z), (\F_r^{-1} U_\lambda^{t}\F_r)_{t \in \Z}, \F_r^{-1} D\F_r)$ of the model quantum walk is identical to 
\[\left( L^2(\T_{2 \pi r^{-1}}), (M[\lambda]^t)_{t \in \Z}, \dfrac{d}{i dk} \right).\]

The model quantum walk was first introduced in \cite{SaigoSako}. The formulation in \cite{SaigoSako} is different from that in this paper. However, the difference is not crucial.
In  \cite{SaigoSako}, the parameter $r$ was a reciprocal $r = d^{-1}$ of a natural number $d$.
A quantum walk
\[\left( \ell_2 (\Z) \otimes \C^d, \left(\F_{d^{-1}} M[\lambda]^t \F_{d^{-1}}^{-1} \right)_{t \in \Z}, D_1 \otimes {\rm id} \right)\]
was called a model quantum walk in \cite{SaigoSako}.
The Hilbert space $\ell_2 (\Z) \otimes \C^d$ can be identified with $\ell_2(d^{-1} \Z)$ by
\[\delta_s \otimes \delta_k \mapsto \delta_{k d^{-1} + s}, \quad s \in \Z, k \in \{1, 2, \cdots, d\}.\]
As explained in Example \ref{example: replacing D},
the self-adjoint operator $D_1 \otimes {\rm id}$ can be replaced with $D_{d^{-1}}$.
It turns out that the model quantum walk in \cite{SaigoSako} is similar to
\[\left( \ell_2 (d^{-1} \Z), \left(\F_{d^{-1}} M[\lambda]^t \F_{d^{-1}}^{-1} \right)_{t \in \Z}, D_{d^{-1}} \right).\]
This is a special case of model quantum walks defined in this paper.

\subsection{Review of a structure theorem in \cite{SaigoSako}}
\label{subsection: review of SS}
We have already obtained a structure theorem for discrete-time homogeneous analytic quantum walks.
We review here the main result in \cite{SaigoSako} and adjust the notations for the argument of this paper.

Let $U = (\ell_2(\Z) \otimes \C^n, (U^t)_{t \in \Z}, D_1 \otimes {\rm id}, S_1 \otimes {\rm id})$ be an arbitrary homogeneous analytic quantum walk.
The inverse Fourier transform of $U$ is
\[\widehat{U} = 
\left(
L^2(\T_{2 \pi}) \otimes \C^n, 
((\F_1 \otimes {\rm id})^{-1} U^t (\F_1 \otimes {\rm id}))_{t \in \Z}, 
\frac{d}{i d k} \otimes {\rm id}
\right).
\]
The generator $\widehat{U}$ is a unitary element of $C(\T_{2 \pi}) \otimes M_n(\C)$, the space of $(n \times n)$-matrices whose entries are multiplication operators given by analytic functions on $\T_{2 \pi}$.
For every $k \in \T_{2 \pi}$, $\widehat{U}$ gives an $(n \times n)$-unitary matrix $\widehat{U}(k)$.
The unitary matrix provides a decomposition of $\C^n$ into eigenspaces.
By analyticity of the entries of $\widehat{U}$, we obtain not only analytic functions $\lambda(k)$ of eigenvalues of $\widehat{U}(k)$, but also analytic sections of eigenvectors whose fibers make orthonormal bases of $\C^n$.
We need to keep in mind that the eigenvalue functions $\lambda(k)$ are not necessarily single-valued.

To describe the multi-valued eigenvalue functions,
we make use of the torus $\T_{2 \pi d} = \R / (2 \pi d \Z)$,
and define a covering map $p_d \colon \T_{2 \pi d} \to \T_{2 \pi} = \R / (2 \pi \Z)$ by the standard quotient.
We obtain 
\begin{itemize}
\item
natural numbers $d(1), d(2), \cdots, d(\nu)$ whose sum is $n$,
\item
analytic maps $\lambda_\iota \colon \T_{2 \pi d(\iota)} \to \T$, $(\iota = 1, 2, \cdots, \nu)$
\end{itemize}
such that for every $k \in \T_{2 \pi}$, the set of the eigenvalues of $\widehat{U}(k)$ is
\[\bigcup_{\iota = 1}^\nu \left\{ \left. \lambda_\iota \left(\widetilde{k} \right) \ \right|\ \widetilde{k} \in \T_{2 \pi d(\iota)}, p_{d(\iota)}\left(\widetilde{k}\right) = k\right\}.\]
Corresponding to this description of eigenvalues, an analytic sections ${\bf v}_\iota  \left( \widetilde{k} \right)$ of eigenvectors do exist.
These sections naturally define a unitary operator
\[V \colon \bigoplus_{\iota = 1}^\nu L^2(\T_{2 \pi d(\iota)}) \to L^2(\T_{2 \pi}\to \C^n) = L^2(\T_{2 \pi}) \otimes \C^n\]
by the formula
\[[V (\xi_\iota)]\left( k \right) = 
\sum_{\widetilde{k},\ p_\iota\left( \widetilde{k} \right) = k}
\xi_\iota \left( \widetilde{k} \right)
{\bf v}_\iota  \left( \widetilde{k} \right), 
\quad \xi_\iota \in L^2(\T_{2 \pi d(\iota)}),
k \in \T_{2 \pi}.\]
By analyticity of the sections of eigenvectors,
$V$ is analytic with respect to 
\[\frac{d}{i d k} \oplus \frac{d}{i d k} \oplus \cdots \oplus \frac{d}{i d k}\ (\nu {\rm-times}) {\rm \ and \ } \frac{d}{i d k} \otimes {\rm id}.\]
The analytic unitary operator $V$ gives similarity between $\widehat{U}$
and the direct sum
\[\bigoplus_{\iota = 1}^\nu 
\left(
L^2(\T_{2 \pi d(\iota)}),
(M[\lambda_\iota]^t)_{t \in \Z},
\frac{d}{i d k}
\right).
\]
Applying Fourier transform, we conclude that $U$ is similar to a direct sum of model quantum walks.

\begin{example}[$4$-state Grover walk]
\label{example: 4-state Grover walk}
Consider the following unitary operator on $\ell_2(\Z) \otimes \C^4 \cong \ell_2(\Z)^4$:
\[U = 
\frac{1}{2}
\left(
\begin{array}{ccccc}
S_1^{-3} &    0    &   0  & 0       \\
0       & S_1^{-1} &   0  & 0 \\
0       &    0     &  S_1  & 0       \\
0       &    0     &  0  & S_1^{3}
\end{array}
\right)
\left(
\begin{array}{ccccc}
- 1 & 1 & 1 & 1 \\
1 & - 1 & 1 & 1 \\
1 & 1 & - 1 & 1 \\
1 & 1 & 1 & - 1 \\
\end{array}
\right).
\]
We concretely calculate the eigenvalue function of $U$ and identify the decomposition into model quantum walks.
The inverse Fourier transform of $U$ is
\[\widehat{U}(k) = 
\frac{1}{2}
\left(
\begin{array}{ccccc}
e^{-3ik} &    0    &   0  & 0       \\
0       & e^{-ik} &   0  & 0 \\
0       &    0     &  e^{ik}  & 0       \\
0       &    0     &  0  & e^{3ik}
\end{array}
\right)
\left(
\begin{array}{ccccc}
- 1 & 1 & 1 & 1 \\
1 & - 1 & 1 & 1 \\
1 & 1 & - 1 & 1 \\
1 & 1 & 1 & - 1 \\
\end{array}
\right),
\quad k \in \T_{2 \pi}.
\]
The characteristic polynomial is
\begin{eqnarray*}
& & \det \left( \lambda - \widehat{U}(k) \right)\\
&=& \lambda^4 + \frac{e^{3ik} + e^{ik} + e^{-ik} + e^{-3ik}}{2} \lambda^3 
- \frac{e^{3ik} + e^{ik} + e^{-ik} + e^{-3ik}}{2} \lambda -1\\
&=& (\lambda - 1) (\lambda + 1)
\left\{ \lambda^2 + (\cos 3k + \cos k) \lambda + 1\right\}.
\end{eqnarray*}
We obtain two constant eigenvalue functions $\lambda_1(k) = 1$ and $\lambda_2(k) = -1$.
We focus on the roots of the last factor.
The roots are given by
\begin{eqnarray*}
\lambda_3(k) &=& 
- \dfrac{\cos k + \cos 3k}{2} - i \sin k \sqrt{1 + 4 \cos^4 k},
\\
\lambda_4(k) &=& 
- \dfrac{\cos k + \cos 3k}{2} + i \sin k \sqrt{1 + 4 \cos^4 k}.
\end{eqnarray*}
Thus we obtain four single-valued analytic eigenvalue functions
$\lambda_1(k) = 1$, $\lambda_2(k) = -1$, $\lambda_3(k)$, $\lambda_4(k)$,
satisfying
\[\det \left( \lambda - \widehat{U}(k) \right) = \prod_{j = 1}^4 (\lambda - \lambda_j(k)).\]
By the structure theorem, it turns out that
the $4$-state Grover walk 
$\left(
\ell_2(\Z) \otimes \C^4, 
(U^t),
D_1 \otimes {\rm id}
\right)
$
is similar to the direct sum:
\begin{eqnarray*}
\left(
L^2(\T_{2 \pi}) \otimes \C^4, 
(1 \oplus (-1)^t \oplus M[\lambda_3]^t \oplus M [\lambda_4]^t)_{t \in \Z},
\left( \frac{d}{id k} \right)^{\oplus 4}
\right)
\end{eqnarray*}
Similarity is given by a composite of the inverse Fourier transform and a unitary operator $\widehat{V} = \left( \widehat{V}(k) \right)_{k \in \T_{2 \pi}}$ acting on $L^2(\T_{2 \pi}) \otimes \C^4$.
The unitary $\widehat{V}(k) \in M_4(\C)$ is given by an analytic decomposition into eigenvectors and therefore analytic with respect to $\left( \frac{d}{i dk} \right)^{\oplus 4}$.
\end{example}

\begin{example}[$3$-state Grover walk]
\label{example: 3-state Grover walk}
Consider the following unitary operator on $\ell_2(\Z) \otimes \C^3 \cong \ell_2(\Z)^3$:
\[U = 
\frac{1}{3}
\left(
\begin{array}{ccccc}
S_1^{-1} & 0 & 0       \\
0       & 1 & 0 \\
0       & 0 &  S_1 \\
\end{array}
\right)
\left(
\begin{array}{ccccc}
- 1 & 2 & 2 \\
2 & - 1 & 2 \\
2 & 2 & - 1 \\
\end{array}
\right).
\]
We concretely calculate the eigenvalue function of $U$ and identify the decomposition into model quantum walks.
The inverse Fourier transform of $U$ is
\[\widehat{U}(k) = 
\frac{1}{3}
\left(
\begin{array}{ccccc}
e^{-ik} &    0    &   0 \\
0       &    1     &   0 \\
0       &    0     &  e^{ik} \\
\end{array}
\right)
\left(
\begin{array}{ccccc}
- 1 & 2 & 2 \\
2 & - 1 & 2 \\
2 & 2 & - 1
\end{array}
\right),
\quad k \in \T_{2 \pi}.
\]
The characteristic polynomial is
\begin{eqnarray*}
& & \det \left( \lambda - \widehat{U}(k) \right)\\
&=& \lambda^3 
+ \frac{e^{ik} + 1 + e^{-ik}}{3} \lambda^2 
- \frac{e^{ik} + 1 + e^{-ik}}{3} \lambda
1\\
&=& (\lambda - 1)
\left( \lambda^2 + \frac{4 + 2 \cos k}{3} \lambda + 1\right).
\end{eqnarray*}
Here we obtain an eigenvalue function $\lambda_1(k) = 1$.
We focus on the roots of the last factor.
For $k \in [0, 2 \pi)$, the roots are given by
\begin{eqnarray*}
\lambda_2(k) &=& 
- \dfrac{2 + \cos k}{3} - \frac{1}{3} i \sin \frac{k}{2} \sqrt{10 + 2 \cos k}
\end{eqnarray*}
and
$\lambda_2 (k + 2 \pi) = \overline{\lambda_2(k)}$.
Thus we obtain one single-valued analytic eigenvalue function $\lambda_1(k) = 1$ and
one multi-valued analytic eigenvalue function 
\[k \mapsto \{\lambda_2(k), \lambda_2(k + 2 \pi)\}.\]
By the structure theorem, 
it turns out that the direct sum
\begin{eqnarray*}
\left(
L^2(\T_{2 \pi}) \oplus L^2(\T_{4 \pi}), 
(1 \oplus M[\lambda_2]^t)_{t \in \Z},
\frac{d}{i dk} \oplus \frac{d}{i dk}
\right)
\end{eqnarray*}
is similar to
the $3$-state Grover walk 
$\left(
\ell_2(\Z) \otimes \C^3, 
(U^t),
D_1 \otimes {\rm id}
\right)
$.
Inui, Konno, and Segawa in \cite{InuiKonnoSegawa} observed that the limit distribution of the $3$-state Grover walk is localized around $0 \in \R$.
The above decomposition gives another proof of their result, since the walk contains a constant quantum walk as a direct summand.

Similarity is given by a composite of 
a unitary operator 
\[\widehat{V} \colon L^2(\T_{2 \pi}) \oplus L^2(\T_{4 \pi}) \to L^2(\T_{2 \pi}) \otimes \C^3\]
and the Fourier transform
\[\F_1 \otimes {\rm id} \colon L^2(\T_{2 \pi}) \otimes \C^3 \to \ell_2(\Z) \otimes \C^3.\]
The unitary $\widehat{V}(k)$ is given by an analytic decomposition into eigenvectors of $\widehat{U}(k)$ corresponding to the eigenvalues $\{1, \lambda_2(k), \lambda_2(k + 2 \pi)\}$.
The unitary $\widehat{V}$ is analytic with respect to $\frac{d}{i dk}$.
\end{example}

Here we propose a problem.
When is a homogeneous analytic quantum walk $U$ realized by a continuous-time uniform quantum walk which is {\it not necessarily homogeneous}?
Theorem \ref{theorem: continuous-time} gives an answer.

Some model quantum walk can be decomposed into a direct sum of model quantum walks.
To analyze such a case, we need to decompose $U$ further.
 
\subsection{Prime model quantum walks}

Let $\lambda$ a periodic analytic map $\lambda \colon \R \to \T$ on $\R$
which is not constant.
For a positive period $2 \pi r^{-1}$ of $\lambda$, $\lambda$ gives an analytic map $\lambda \colon \T_{2 \pi r^{-1}} \to \T$.
We can construct model quantum walks from $\lambda$.
Since the period $2 \pi r^{-1}$ of $\lambda$ is not unique, the model quantum walk
$(\ell_2(r \Z), (U_\lambda^{t})_{t \in \Z}, D)$ is not uniquely determined by $\lambda$.
But the possible model quantum walks are closely related to that of the minimal period.

We note that the dual group of $(r / m) \Z$ is $\T_{2 \pi m r^{-1}} = \R / (2 \pi m r^{-1} \Z)$.
The length of $\T_{2 \pi m r^{-1}}$ is $m$-times longer than that of $\T_{2 \pi r^{-1}}$.

\begin{proposition}\label{proposition: decomposition into prime walks}
Let $\lambda$ be a periodic analytic map $\lambda \colon \R \to \T$ on $\R$ which is not a constant.
Let $2 \pi r^{-1}$ be the minimal period of $\lambda$.
Let $m$ be a natural number.
The model quantum walk
$\left( \ell_2\left( \dfrac{r}{m} \Z \right), (U_\lambda^{t})_{t \in \Z}, D_{r / m} \right)$
is similar to the direct sum
\[\left( \ell_2\left( r \Z \right), (U_\lambda^{t})_{t \in \Z}, D_r \right)^{\oplus m}\]
 of the model quantum walks given by the minimal period.
\end{proposition}

\begin{remark}
For the definition of similarity, see Definition \ref{definition: similarity}.
Similarity as homogeneous quantum walks defined in Definition \ref
{definition: homogeneous}
does not necessarily hold.
\end{remark}

\begin{proof}
It suffices to show that
\[\left( L^2\left( \R / (2 \pi m r^{-1} \Z )\right), (M[\lambda]^t)_{t \in \Z}, \frac{d}{i d k} \right)\]
is similar to
\[\left(L^2\left( \R / (2 \pi r^{-1} \Z )\right), (M[\lambda]^t)_{t \in \Z}, \frac{d}{i d k} \right)^{\oplus m}\]

Recall that for $x \in (r / m)\Z$, $c_x(k) = \exp(i k x)$ defines a character of $\T_{2 \pi m r^{-1}} = \R / (2 \pi m r^{-1} \Z )$ and that
$\{c_x \ |\ x \in (r / m)\Z\}$ is an orthonormal basis of $L^2 (\T_{2 \pi m r^{-1}})$.
Since the minimal period of $k$ is $2 \pi r^{-1}$,
we can express the analytic map $\lambda$ by
\[\lambda(k) = \sum_{x \in r \Z} \alpha_x c_x.\]
Since $\lambda$ is analytic, as $|x| \to \infty$, $|\alpha_x|$ rapidly decreases.
Therefore the infinite sum
\[M[\lambda] = \sum_{x \in r \Z} \alpha_x M[c_x] \in \B(L^2(\T_{2 \pi m r^{-1}}))\]
converges in the operator norm topology.
By the relation $M[c_x](c_y) = c_{x + y}$,
for every $z \in \{0, r/m,  2r/m, \cdots, (m - 1) r/m\}$,
\[\Hil_z := \mathrm{\overline{span}} \{c_{x + z} \ |\ x \in r\Z\}\]
is invariant under the action of $M[\lambda]$ and under $\exp(i k \frac{d}{i dk})$.
The action of $M[\lambda]$ on $\Hil_z$ is unitary equivalent to that on $\Hil_0$.
The operator $\frac{d}{idk}$ on $\Hil_z$ corresponds to the sum of a constant operator and $\frac{d}{idk}$ on $\Hil_0$.
The Hilbert space $\Hil_0 = \mathrm{\overline{span}} \{c_{x} \ |\ x \in r\Z\}$ is naturally identified with $L^2(\T_{2 \pi r^{-1}})$. 
It follows that
\[\bigoplus_z \left(\Hil_z, (M[\lambda]^t)_{t \in \Z}, \frac{d}{i d k} \right)\]
and
\[\left(L^2\left( \R / (2 \pi r^{-1} \Z )\right), (M[\lambda]^t)_{t \in \Z}, \frac{d}{i d k} \right)^{\oplus m}\]
are similar.
The former walk is unitary equivalent to the original quantum walk.
\end{proof}

\begin{example}
Consider the quantum walk $(\ell_2(\Z) \otimes \C^3, (U^t)_{t \in \Z}, D_1 \otimes \mathrm{id})$ generated by
$U = 
\left(
\begin{array}{ccc}
0 & S_1 & 0 \\
0 & 0 & S_1 \\
1 & 0 & 0
\end{array}
\right)$.
The inverse Fourier transform is
$\widehat{U}(k) = 
\left(
\begin{array}{ccc}
0 & e^{ik} & 0 \\
0 & 0 & e^{ik} \\
1 & 0 & 0
\end{array}
\right), k \in \T_{2 \pi}$.
The characteristic polynomial is $\lambda^3 - e^{2 i k}$.
The eigenvalue function is a multi-valued function
\[k \mapsto \{\lambda_1(k), \lambda_1(k + 2 \pi), \lambda_1(k + 4 \pi)\}\]
given by
\[c_{2/3} (k) = \exp(2 i k/3), \quad k \in 6 \pi.\]
By the structure theorem in \cite{SaigoSako},
the original quantum walk is similar to
\[\left( L^2(\T_{6 \pi}), \left( M[c_{2/3}]^t \right)_{t \in \Z}, \frac{d}{i d k} \right).\]
The function $\lambda_1 \colon \T_{6 \pi} \to \T$ has a non-trivial period $3 \pi$.
As in Proposition \ref{proposition: decomposition into prime walks},
This is similar to the following direct sum:
\[\left( L^2(\T_{3 \pi}), \left( M[c_{2/3}]^t \right)_{t \in \Z}, \frac{d}{i d k} \right)^{\oplus 2}.\]
This is unitary equivalent to the direct sum of two prime model quantum walks
\[\left( \ell_2 \left( (2 / 3) \Z \right), \left( M[c_{2/3}]^t \right)_{t \in \Z}, D_{2 / 3} \right)^{\oplus 2}.\]
Note that the original quantum walk $U$ and this direct sum are similar in the category of quantum walks, but not similar in the category of homogeneous quantum walks (Subsection \ref{subsection: homogeneous}).

\end{example}

\begin{definition}
A model quantum walk $\left( \ell_2\left( r \Z \right), (U_\lambda^{t})_{t \in \Z}, D \right)$
is said to be {\rm prime}, if the analytic map $\lambda \colon \T_{2 \pi r^{-1}} \to \T$ has no period other than $0$. 
\end{definition}

Thus we have the following structure theorem.

\begin{theorem}\label{theorem: structure theorem}
Let $U = (\ell_2(\Z) \otimes \C^n, (U^t)_{t \in \Z}, D_1 \otimes {\rm id})$ be an arbitrary one-dimensional discrete-time homogeneous analytic quantum walk.
Then there exist 
\begin{itemize}
\item
non-negative integers $l$, $m$, 
\item
rational numbers $r(j)$, $(j \in \{1, \cdots, l\})$,
\item
prime model quantum walks
$\left( \ell_2(r(j) \Z), \left( U_{\lambda(j)}^t \right)_{t \in \Z}, D_{r(j)} \right)$, $(j \in \{1, \cdots, l\})$,
\item
complex numbers $\alpha(k)$, $(k \in \{1, \cdots, m\})$ whose absolute values are $1$,
\end{itemize}
satisfying 
\begin{itemize}
\item
that the given analytic walk
$U$
is similar to the direct sum
\[
\bigoplus_{j = 1}^l 
\left( \ell_2(r(j) \Z), \left( U_{\lambda(j)}^t \right)_{t \in \Z}, D_{r(j)} \right)
\oplus
\bigoplus_{k = 1}^m 
\left( \ell_2(\Z), \left( \alpha(k)^t \right)_{t \in \Z}, D_1 \right)
\]
(The integers $l$ and $m$ can be zero. In the case that $l = 0$, erase the first half. In the case that $m = 0$, erase the second half.)
\item
and that the degree of freedom $n$ is equal to
$m + \sum_{j = 1}^l 
r(j)^{-1}.$
\end{itemize}
\end{theorem}

\begin{proof}
As we explained in Subsection \ref{subsection: review of SS},
$U$ is similar to the direct sum of model quantum walks
\[
\bigoplus_{\iota = 1}^\nu
\left( \ell_2(r(\iota) \Z), \left( U_{\lambda(\iota)}^t \right)_{t \in \Z}, D_{r(\iota)} \right).
\]
The positive numbers $r(\iota)$ are reciprocals of natural numbers $d(\iota)$.

Consider the case that the analytic function $\lambda(\iota) \colon \T_{2 \pi d(\iota)} \to \T$ is not constant.
If it is not prime, we can further decompose the
the model quantum walk 
into prime model quantum walks.
In such a case, $r(\iota)$ becomes larger, but the sum of reciprocals is preserved (see Proposition \ref{proposition: decomposition into prime walks}).
Each prime model quantum walk becomes a direct summand of the first half.

If the analytic function $\lambda(\iota) \colon \T_{2 \pi d(\iota)} \to \T$ is a constant function $\alpha(\iota)$,
then the corresponding model quantum walk is decomposed as follows:
\[\left( \ell_2(r(\iota) \Z), \left( \alpha(\iota)^t \right)_{t \in \Z}, D_{r(\iota)} \right)
\cong
\left( \ell_2(\Z), \left( \alpha(\iota)^t \right)_{t \in \Z}, D_1 \right)^{\oplus r(\iota)^{-1}}.
\]
These direct summands becomes direct summands of the second half.
\end{proof}

Examples 
\ref{example: 4-state Grover walk}
and
\ref{example: 3-state Grover walk}
give decompositions into constant quantum walks and prime model quantum walks.

\subsection{Uniform intertwiner between two walks}
\label{subsection: uniform intertwiner}

For a while, we make use of two pairs of dual groups $(r(1)\Z, \T_{l(1)})$ and $(r(2)\Z, \T_{l(2)})$.
As explained in Subsection \ref{subsection: Fourier analysis},
for $\iota = 1, 2$, the length $l(\iota)$ of the torus $\T_{l(\iota)}$ is equal to $2 \pi r(\iota)^{-1}$.

The following lemma is the most important technical ingredient of this paper.

\begin{lemma}\label{lemma: diffeomorphism on T}
Let $\psi \colon \T_{l(2)} \to \T_{l(1)}$ be a diffeomorphism.
Let $f$ be a bounded Borel function $\T_{l(2)}$ whose support is not null.
Let $\widehat{V}$ be the composite $M[f] \circ \psi^*$ of
\begin{itemize}
\item
the pull back $\psi^* \colon L^2(\T_{l(1)}) \to L^2(\T_{l(2)})$ of $\psi$ and
\item
the multiplication operator $M[f] \colon L^2(\T_{l(2)}) \to L^2(\T_{l(2)})$.
\end{itemize}
If $\widehat{V} \in \mathcal{B}(L^2(\T_{l(2)}) \leftarrow L^2(\T_{l(1)}))$ is uniform with respect to the differential operators $\frac{d}{i d k}$ on $\T_{l(1)}$ and $\frac{d}{i d k}$ on $\T_{l(2)}$,
then on every interval contained in ${\rm supp} f$, $\psi(k) - k$ is constant.
\end{lemma}

\begin{proof}
For $y \in r(1) \Z$, let $c_y$ denote the function on $\T_{l(1)} = \R / (r(1) \Z)$ given by $c_y(k) = \exp(i k y)$.
For $x \in r(2) \Z$, let $c_x$ denote the function on $\T_{l(2)} = \R / (r(2) \Z)$ given by $c_x(k) = \exp(i k x)$.
For $x \in r(2) \Z$ and $y \in r(1) \Z$, the matrix coefficient $V_{x, y}$ of the Fourier transform of $\widehat{V}$ is given by
\begin{eqnarray*}
V_{x, y} = \left\langle M[f] \psi^* c_y, c_x \right\rangle_{L^2 \left( \T_{l(2)} \right)}
= \int_{0}^{l(2)} \exp(-i k x) f \left( k \right) \exp(i \psi \left(k \right) y)
\cdot \dfrac{dk}{l(2)}.
\end{eqnarray*}
Define a function $\Psi \colon \T_{l(2)} \to \T = \{z \in \C | \ |z| = 1\}$ by
$\Psi(k) = \exp(i \psi(k) r(1))$.
The above quantity is equal to
\begin{eqnarray}\label{equation: a new quantum walk}
V_{x, y} &=& \left\langle M \left[\Psi^{y / r(1)} \right] f, c_x \right\rangle_{L^2 \left( \T_{l(2)} \right)}.
\end{eqnarray}

Motivated by the above formula, we consider the homogeneous smooth quantum walk $(L^2(\T_{l(2)}), (M[\Psi^{y / r(1)}])_{y \in r(1) \Z}, \frac{d}{i d k})$ and the initial vector $f \in L^2(\T_{l(2)})$.
For the rest of this proof, consider the case that $y \in r(1) \Z$ is large, and we regard the integer $y / r(1)$ as time. Define an integer $t(y)$ by $y / r(1)$.
The Fourier coefficients $(V_{x, y})_{x} \in \ell_2(r(2) \Z)$ of  $M[\Psi^{t(y)}] f \in L^2 (\T_{l(2)})$ gives a measure on $r(2) \Z$.
Denote by $p_y$ the push-forward measure along the mapping $r(2) \Z \ni x \to x / y \in \R$.
More precisely, $p_y$ is the sum $\sum_{x \in r(2) \Z} |V_{x, y}|^2 \delta_{x / y}$ of point masses
at $\{x/y \ |\ x \in r(2) \Z\}$.

We first consider the case that $f \colon \T_{l(2)} \to \C$ is smooth.
Denote by $D = \frac{d}{i dk}$ the differential operator acting on $L^2 (\T_{l(2)})$.
By Lemma \ref{lemma: moment},
the $m$-th moment of $p_y$ is identical to
\begin{eqnarray*}
& &
\left\langle \left(\dfrac{D}{y}\right)^m M\left[\Psi^{t(y)}\right] f, M\left[\Psi^{- t(y)}\right] f \right\rangle_{L^2 \left( \T_{l(2)} \right)}\\
&=&
\left\langle \left( M\left[\Psi^{- t(y)}\right] \dfrac{D}{y}  M\left[\Psi^{t(y)}\right] \right)^m f, f \right\rangle_{L^2 \left( \T_{l(2)} \right)}\\
&=&
\left\langle \left( M\left[\Psi^{- t(y)} \dfrac{t(y)}{i y} \Psi' \Psi^{t(y) -1}\right] + \dfrac{D}{y} \right)^m f, f \right\rangle_{L^2 \left( \T_{l(2)} \right)}\\
&=&
\left\langle \left( M\left[\frac{\Psi'}{i r(1) \Psi}\right] + \dfrac{D}{y} \right)^m f, f \right\rangle_{L^2 \left( \T_{l(2)} \right)}.
\end{eqnarray*}
The function $\Psi' / (i r(1) \Psi)$ is equal to $\psi'$.
As $y \to \infty$, the moment of $p_y$ tends to
\[\left\langle M[\psi']^m f, f \right\rangle_{L^2 \left( \T_{l(2)} \right)}
= \int_0 ^{l(2)} \psi'(k)^m |f(k)|^2 \cdot \dfrac{dk}{l(2)}.
\]
This implies that $p_y$ converges in law to the push-forward of the measure $|f(k)|^2 \dfrac{dk}{l(2)}$ along the mapping $\psi' \colon \T_{l(2)} \to \R$.
It follows that $p_y$ weakly converges to the push-forward measure
(Proposition \ref{proposition: two kinds of convergence}).
The vector $f$ is not a unit vector, but the argument in Proposition \ref{proposition: two kinds of convergence} is valid.

Let us go back to the general case.
Suppose that $g \colon \T_{l(2)} \to \C$ is a general bounded Borel function.
Denote by $\widehat{W}$ the operator $M \left[g \right] \circ \psi^*$.
The matrix coefficients $W_{x, y} = \left\langle \widehat{W} c_y, c_x \right\rangle_{L^2(\T_{l(2)})}$ of the Fourier transform of $\widehat{W}$ are given by
\begin{eqnarray*}
W_{x, y} &=& \left\langle M \left[\Psi^{t(y)} \right] g, c_x \right\rangle_{L^2(\T_{l(2)})}.
\end{eqnarray*}
Let $q_y$ be the probability measure $\sum_{x \in r(2) \Z} \left|W_{x, y} \right|^2 \delta_{x / y}$.
Define a constant $C$ by 
$\left\| g \right\|_{L^2 (\T_{l(2)})}$.
For an arbitrary positive number $\epsilon$, there exist a smooth function $f \colon \T_{l(2)} \to \C$ satisfying
$\left\|f - g \right\|_{L^2(\T_{l(2)})} < \epsilon$, 
$\left\|f \right\|_{L^2(\T_{l(2)})} \le C$.
We denote by $\| \cdot \|_{{\rm cb}^*}$ the norm of linear functionals on the Banach space of bounded continuous functions on $\R$.
By the Cauchy–-Schwarz inequality, we have
\begin{eqnarray*}
\left\| p_y - q_y  \right\|_{{\rm cb}^*}
&\le& 
\left\| \left(|V_{x, y}|^2 \right)_{x} - \left( \left| W_{x, y} \right|^2  \right)_{x}  \right\|_{\ell_1}\\
&\le& 
\left\| (V_{x, y})_{x} + \left(W_{x, y} \right)_{x}  \right\|_{\ell_2}
\left\| (V_{x, y})_{x} - \left(W_{x, y} \right)_{x}  \right\|_{\ell_2}.
\end{eqnarray*}
By the Plancherel theorem, and by the equation $|\Psi(k)| = 1$, we have
\begin{eqnarray*}
& & 
\left\| p_y - q_y  \right\|_{{\rm cb}^*}\\
&\le& 
\left\| M\left[\Psi^{t(y)} \right] f + M\left[\Psi^{t(y)} \right] g \right\|_{L^2(\T_{l(2)})}
\left\| M\left[\Psi^{t(y)} \right] f - M\left[\Psi^{t(y)} \right] g \right\|_{L^2(\T_{l(2)})}\\
&=& 
\left\| f + g \right\|_{L^2(\T_{l(2)})}
\left\| f - g \right\|_{L^2(\T_{l(2)})}
\\
&\le&
2C \epsilon. 
\end{eqnarray*}
By the Cauchy–-Schwarz inequality, we have
\begin{eqnarray*}
& &
\left\| \psi'_* \left( |f(k)|^2 \dfrac{dk}{l(2)} \right)
 - \psi'_* \left( \left| g(k) \right|^2 \dfrac{dk}{l(2)} \right)
 \right\|_{{\rm cb}^*}\\
&\le&
\left\| |f(k)|^2 \dfrac{dk}{l(2)}
 - 
\left| g(k) \right|^2 \dfrac{dk}{l(2)}
 \right\|_{{\rm cb}^* (\T_{l(2)})}\\
&=&
\left\| |f|^2 
 - 
\left| g \right|^2 \right\|_{L^1(\T_{l(2)})}\\
&\le& 
\left\| f + g \right\|_{L^2(\T_{l(2)})}
\left\| f - g \right\|_{L^2(\T_{l(2)})}\\
&\le&
2C \epsilon.
\end{eqnarray*}
Because $p_y$ weakly converges to $\psi'_* \left( |f(k)|^2 \dfrac{dk}{l(2)} \right)$,
$\widetilde{p}_y$ weakly converges to $\psi'_* \left( \left| g(k) \right|^2 \dfrac{dk}{l(2)} \right)$.

By Proposition \ref{proposition: converging to an atom},
since the Fourier transform $W$ of $\widehat{W}$ is uniform with respect to the diagonal operators $D_{r(1)}$ and $D_{r(2)}$,
the weak limit of $q_y$ has to be concentrated on $1$.
It follows that $\psi'$ is the constant function $1$ on the support of $g$.
\end{proof}

\begin{lemma}\label{lemma: diffeomorphism between intervals}
Let $\widehat{V}$ be a bounded operator from $L^2(\T_{l(1)})$ to $L^2(\T_{l(2)})$.
Let $I(1)$ be an open interval of $\T_{l(1)}$ and
let $I(2)$ be an open interval of $\T_{l(2)}$.
Let $\psi \colon I(2) \to I(1)$ be a diffeomorphism.
Suppose that the derivatives $\psi'$ and $\left( \psi^{-1} \right)'$ are bounded.
Let $f \colon I(2) \to \C$ be a bounded Borel non-zero function.
Suppose that the restriction $\widehat{V} |_{L^2(I(1))}$ is identical to the composition operator  $M[f] \circ \psi^*$ of
\begin{itemize}
\item
the pull back $\psi^* \colon L^2(I(1)) \to L^2(I(2))$ of $\psi$ and
\item
the multiplication operator $M[f] \colon L^2(I(2)) \to L^2(I(2)) \subset L^2(\mathbb{T})$.
\end{itemize}
If $\widehat{V}$ is uniform with respect to the differential operators $\frac{d}{i dk}$ on $L^2(\T_{l(1)})$ and $\frac{d}{i dk}$ on $L^2(\T_{l(2)})$,
then on every interval contained in ${\rm supp} f \cap I(2)$, $\psi(k) - k$ is constant.
\end{lemma}

\begin{proof}
Let $g$ be an arbitrary smooth function on $\T_{l(1)}$ such that the support ${\rm supp} (g)$ is a compact subset of $I(1)$.
The multiplication operator $M[g]$ maps $L^2(\T_{l(1)})$ to $L^2({\rm supp} (g)) \subset L^2(I(1))$.
Since $M[g]$ is uniform with respect to $\frac{d}{id k}$, $\widehat{V} M[g]$ is uniform with respect to $\frac{d}{id k}$. The operator $\widehat{V} M[g]$ expressed as follows:
\[\widehat{V} M[g]
= M[f] \circ \psi^* \circ M[g].\]
Choose a diffeomorphism $\phi \colon \T_{l(2)} \to \T_{l(1)}$ which is identical to $\psi$ on $\psi^{-1} ({\rm supp} (g))$.
Then we have
\[\widehat{V} M[g] 
= M[f] \circ \phi^* \circ M[g]
= M[f \cdot (g \circ \phi)] \circ \phi^*.\]
By Lemma \ref{lemma: diffeomorphism on T},
$\phi$ is rotation on every interval included in ${\rm supp}f \cap {\rm supp}(g \circ \phi)$, and therefore, $\psi$ is rotation on every interval included in ${\rm supp}f \cap {\rm supp}(g \circ \psi)$.
It follows that for every interval included in ${\rm supp} f \cap I(2)$, the map $\psi$ is rotation.
\end{proof}

\begin{proposition}\label{proposition: uniform intertwiner}
Let $\lambda_1 \colon \T_{l(1)} \to \T$ and $\lambda_2 \colon \T_{l(2)} \to \T$ be analytic maps. 
Assume that $\lambda_1$ and $\lambda_2$ do not have period.
Let $\left( \ell_2 \left( r(1) \Z \right), (U_1^{t})_{t \in \Z}, D_{r(1)} \right)$ and $( \ell_2\left( r(2) \Z \right)$, $(U_2^{t})_{t \in \Z}$, $D_{r(2)} )$ be the prime model quantum walks given by $\lambda_1$ and $\lambda_2$.
Assume that there exists a non-zero uniform intertwiner between them. Then $l(1)$ is equal to $l := l(2)$, and therefore $r(1)$ is equal to $r := r(2)$. 
There exists (unique) $\alpha \in \T_l$ such that 
\[\lambda_2 (k) = \lambda_1(k + \alpha), \quad k \in \T_l.\]
The set of all the uniform intertwiners
\[\left\{ V \colon \ell_2\left( r \Z \right) \to \ell_2\left( r \Z \right) \ \left| \
V U_1 = U_2 V, 
k \mapsto e^{ik D_r} V e^{-i k D_r} {\rm \ is\ continuous} \right. \right\}\]
is equal to
$\left\{\left. \F_r M[\rho] \F_r^{-1} \circ \exp( i \alpha D_r) \ \right|\ \rho \colon \T_l \to \C {\rm\ continuous} \right\}$.
\end{proposition}

In the proof, for a Borel subset $B \subset \T_{l(\iota)}$, we denote by $1_B$ the definition function of $B$. Note that the multiplication operator $M[1_B]$ is the orthogonal projection $L^2(\T_{l(\iota)}) \to L^2(B)$.

\begin{proof}
Let $\widehat{V} \colon L^2(\T_{l(1)}) \to L^2(\T_{l(2)})$ be a non-zero uniform intertwiner between
\[\left( L^2(\T_{l(1)}), (M[\lambda_1]^t), \frac{d}{i d k} \right), \quad \left( L^2(\T_{l(2)}), (M[\lambda_2]^t), \frac{d}{i d k} \right).\]
For a Borel subset $J \subset \T$, the spectral projection $E_1(J)$ of $M[\lambda_1]$
is the orthogonal projection
\[E_1(J) \colon L^2(\T_{l(1)}) \to L^2( \lambda_1^{-1}(J) ) \subset L^2(\T_{l(1)}).\] 
The spectral projection $E_2(J)$ of $M[\lambda_2]$
is the orthogonal projection
\[E_2(J) \colon L^2(\T_{l(2)}) \to L^2( \lambda_2^{-1}(J) ) \subset L^2(\T_{l(2)}).\] 
The equation $\widehat{V} M[\lambda_1] = M[\lambda_2] \widehat{V}$ implies 
$\widehat{V} E_1(J) = E_2(J) \widehat{V}$.

Since $\lambda_1$ and $\lambda_2$ are not constant function,
by the identity theorem of analytic functions,
the inverse images of a singleton in $\T$ with respect to $\lambda_1$ and $\lambda_2$
is at most finite.
Therefore, the operators $M[\lambda_1]$ and $M[\lambda_2]$ do not have point spectrum.
We also note that the number of critical values of $\lambda_1$ and $\lambda_2$ is finite.
It follows that there exists an open interval $J \subset \T$ satisfying the following:
\begin{itemize}
\item
The interval $J$ does not contain the critical values of $\lambda_1$ nor 
those of $\lambda_2$.
\item
The operator $V E_1(J)$ is not zero. (Therefore $E_2(J) V$ is not zero.)
\end{itemize}
For $\iota = 1, 2$, $\lambda_\iota^{-1}(J)$ consists of
finitely many open intervals.
Note that the restriction of $\lambda_\iota$ on each connected component of $\lambda_\iota^{-1}(J)$ is diffeomorphism onto $J$.
Choose connected components $I(\iota) \subset \lambda_\iota^{-1}(J)$ such that
$M[1_{I(2)}] \widehat{V} M[1_{I(1)}] \neq 0$.
There exist smooth functions $g_\iota$ on $\T_{l(\iota)}$ such that the support of $g_\iota$ is included in $I(\iota)$ and \[M[g_2] \widehat{V} M[g_1] = M[g_2] M[1_{I(2)}] \widehat{V} M[1_{I(1)}] M[g_1] \neq 0.\] 
The operator $M[g_2] \widehat{V} M[g_1]$ is also a uniform intertwiner between $M[\lambda_1]$ and $M[\lambda_2]$.

Replace $\widehat{V}$ with $M[g_2] \widehat{V} M[g_1]$.
Thus we have a non-zero uniform intertwiner $\widehat{V}$ between $M[\lambda_1]$ and $M[\lambda_2]$ and intervals $I(\iota) \subset \T_{l(\iota)}$, $J \subset \T$ satisfying the following
\begin{itemize}
\item
$\lambda_\iota | _{I(\iota)}$ are diffeomorphisms onto $J$,
\item
There exist closed intervals $K(\iota) \subset I(\iota)$ such that 
$\widehat{V} = M[1_{K(2)}] \widehat{V} M[1_{K(1)}]$.
\item
$\widehat{V}$ is uniform with respect to the differential operators $\frac{d}{i d k}$ on $\T_{l(1)}$ and $\T_{l(2)}$.
\end{itemize}
Denote by $\psi$ the diffeomorphism $(\lambda_1 | _{I(1)})^{-1} \circ \lambda_2 | _{I(2)} \colon I(2) \to I(1)$.
Note that $\psi'$ and $(\psi^{-1})'$ are bounded on $K(2)$ and on $K(1)$.

Using a bounded Borel function $g$ on $J$, we can express an arbitrary bounded Borel function on $I(1)$ as $(g \circ \lambda_1) 1_{I(1)}$.
The image of  $(g \circ \lambda_1) 1_{I(1)}$ through $\widehat{V}$ is
\begin{eqnarray*}
\widehat{V}((g \circ \lambda_1) 1_{I(1)}) 
=
\widehat{V} \circ M[g \circ \lambda_1] (1_{I(1)}).
\end{eqnarray*}
Since the operator $M[g \circ \lambda_1]$ is equal to the functional calculus
$\int_{t \in \T} g(t) E_1(d t)$ of $M[\lambda_1]$, we have
\begin{eqnarray*}
\widehat{V}((g \circ \lambda_1) 1_{I(1)}) 
=
\widehat{V} \circ \left(\int_{t \in \T} g(t) E_1(d t) \right) (1_{I(1)}).
\end{eqnarray*}
Define $f \in L^2(\T_{l(2)})$ by $\widehat{V}(1_{I(1)})$.
Because $\widehat{V}$ is an intertwiner between $M[\lambda_1]$ and  $M[\lambda_2]$,
we have
\begin{eqnarray*}
\widehat{V}((g \circ \lambda_1) 1_{I(1)}) 
&=& 
\left(\int_{t \in \T} g(t) E_2(d t) \right) \circ \widehat{V} (1_{I(1)})\\
&=& 
\left(\int_{t \in \T} g(t) E_2(d t) \right) (f)\\
&=& 
M[g \circ \lambda_2] (f)\\
&=&
M[f] ((g \circ \lambda_2) 1_{I(2)}).
\end{eqnarray*}
The function $(g \circ \lambda_2) 1_{I(2)}$ is equal to $((g \circ \lambda_1) 1_{I(1)}) \circ \psi$.
It follows that $\widehat{V} = M[f] \circ \psi^*$.

The function $f = \widehat{V}(1_{I(1)})$ is continuous.
Indeed, we can express $f$ as $\widehat{V}(g_3)$, using a continuous function $g_3$ on $I(1)$ such that $\mathrm{supp}(g_3)$ is included in $I(1)$ and that $g_3(k) = 1$ for $k \in K(1)$.
Since $\widehat V$ is uniform and $g_3$ is continuous, $f = \widehat{V}(1_{I(1)})$ is a continuous on $T_{l(2)}$.

Since $\widehat{V} = M[f] \circ 
\psi^*$ is not zero, ${\rm supp} f$ has to contain an open interval.
By Lemma \ref{lemma: diffeomorphism between intervals},
the mapping $\psi |_{{\rm supp} f}$ is given by rotation on the open interval.
By the identity theorem of analytic functions, $l(1) = l(2)$ and 
there exists $\alpha \in \R$ such that
$\lambda_2(k) = \lambda_1(k + \alpha)$ for every $k \in \T_{l}$.
Note that $\alpha$ is uniquely determined only by $\lambda_1$ and $\lambda_2$, because $\lambda_2$ does not have period.
Define $l$ by $l(1)$ and $r$ by $2 \pi / l$.

To identify the set of all the uniform intertwiners, take an arbitrary uniform intertwiner $\widehat{V}$ between $M[\lambda_1], M[\lambda_2]$. 
There exists finite open intervals
\[J(1), J(2), \cdots, J(\nu) \subset \T_l\]
such that the union $\cup_{\sigma} J(\sigma)$ is the complement of the set of critical values of $\lambda_1$.
It follows that the union $\cup_{\sigma} (J(\sigma) - \alpha)$ is the complement of the set of critical values of $\lambda_2$.
We note that if $\tau \neq \sigma$, then the intertwiner
$M[1_{J(\tau) - \alpha}] \widehat{V} M[1_{J(\sigma)}]$ is zero.
Indeed,
there exist no open intervals $I(2) \subset J(\tau) - \alpha$ and $I(1) \subset J(\sigma)$ such that $\psi = (\lambda_1 | _{I(1)})^{-1} \circ \lambda_2 | _{I(2)} \colon I(2) \to I(1)$ is well-defined and that $\psi$ is rotation.
By the contrapositive of the last paragraph, the intertwiner $M[1_{J(\tau) - \alpha}] \widehat{V} M[1_{J(\sigma)}]$ is zero.

Thus we can express $\widehat{V}$ as follows:
\[\widehat{V} = \sum_{\sigma = 1}^\nu \widehat{V} M[1_{J(\sigma)}]
= \sum_{\sigma = 1}^\nu M[1_{J(\sigma) - \alpha}] \widehat{V} M[1_{J(\sigma)}].\]
For the corner $M[1_{J(\sigma) - \alpha}] \widehat{V} M[1_{J(\sigma)}]$ of $\widehat{V}$,
there exists a Borel function $f_\sigma \colon (J(\sigma) - \alpha) \to \C$ such that
\[M[1_{J(\sigma) - \alpha}] \widehat{V} M[1_{J(\sigma)}] = M[f_\sigma] \circ \phi_\alpha^*,\]
where $\phi_\alpha$ is the rotation by $\alpha \in \T_{2 \pi r^{-1}}$.
Define a Borel function $f$ on $\T_l$, combining $f_\sigma$. We obtain that
\[\widehat{V} = M[f] \circ \phi_\alpha^*.\]
Because $\widehat{V}$ is uniform, $f \colon \T_l \to \C$ has to be continuous.
Therefore the set of uniform intertwiners is included in 
\[\{M[f] \circ \phi_\alpha^* \ |\ f \colon  \T_l \to \C {\rm\ continuous.}\}\]
We can easily show the converse inclusion by direct computation.
The set of Fourier transforms of these operators is nothing other than the set in the proposition.
\end{proof}

Now we are ready to identify the set of uniform intertwiners between given two homogeneous analytic quantum walks.
Theorem \ref{theorem: structure theorem} means that every homogeneous analytic quantum walk is a direct sum of finitely many prime model quantum walks and constant quantum walks.
It suffices to identify the set of uniform intertwiners between these building blocks $U_1$, $U_2$, $\cdots$.

\paragraph{\bf Case 1}
First consider the case that
$U_1$ is a constant quantum walk
$(\ell_2(\Z), (\alpha^t)_t, D_1)$.
and 
that
$U_2$ is a prime model quantum walk
$\left( \ell_2(r \Z), (U_{\lambda}^t)_t, D_{r} \right)$.
Let $V$ be an intertwiner between $U_1$ and $U_2$.
then we have
\[U_\lambda V = V \alpha = \alpha V.\]
Because $U_\lambda$ has no eigenvector other than the zero vector,
$V$ has to be $0$.

\paragraph{\bf Case 2}
Consider the case that
$U_1$ and $U_2$ are constant quantum walks.
We express them as follows
\[
U_1 = \left( \ell_2(\Z), (\alpha^t)_t, D_1 \right), 
\quad 
U_2 = \left( \ell_2(\Z), (\beta^t)_t, D_1 \right). 
\]
If $\alpha \neq \beta$, then there exists no non-zero intertwiner between them.
If $\alpha = \beta$, then every operator is an intertwiner between them.
The collection of the uniform operators is the uniform Roe algebra ${\rm C}^*_{\rm u}(\Z)$. See Remark \ref{remark: uniform Roe algebra}.

\paragraph{\bf Case 3}
Consider the case that $U_1$ and $U_2$ are prime model quantum walks.
We express the quantum walks as follows
\[
U_1 = \left( \ell_2(r(1) \Z), (U_{\lambda_1}^t)_t, D_{r(1)} \right), 
\quad 
U_2 = \left( \ell_2(r(2) \Z), (U_{\lambda_2}^t)_t, D_{r(2)} \right), 
\]
By Proposition \ref{proposition: uniform intertwiner},
if $r: = r(1) = r(2)$, and
if there exists $\alpha \in [0, 2 \pi r^{-1})$ such that $\lambda_2(k) = \lambda_1(k + \alpha)$,
then the set of intertwiners is
\[\left\{\left. \F_r M[\rho] \F_r^{-1} \exp (i\alpha D_r) \ \right|\ \rho \colon \T_{2 \pi r^{-1}} \to \C {\rm\ continuous} \right\}.\]
By Proposition \ref{proposition: uniform intertwiner},
if $r(1) \neq r(2)$, or if there does not exist $\alpha \in \T_{2 \pi r^{-1}}$ such that $\lambda_2(k) = \lambda_1(k + \alpha)$,
then there exists no non-zero uniform intertwiner between $U_1$ and $U_2$.

\begin{example}\label{example: two direct summands of 4-state Grover walk}
Consider the quantum walk generated by
\[W = 
\frac{1}{2}
\left(
\begin{array}{cc}
- S_1^3 - S_1 & S_1 - S_1^{-1}\\
S_1 - S_1^{-1} & - S_1^{-1} - S_1^{-3}
\end{array}
\right) \in \B(\ell_2(\Z) \otimes \C^2).
\]
The characteristic polynomial of the inverse Fourier transform $\widehat{W}(k)$
is
\[\lambda^2 + (\cos 3k + \cos k) \lambda + 1.\]
The roots are given by $\lambda_3(k)$ and $\lambda_4(k)$ defined in Example \ref{example: 4-state Grover walk}.
The quantum walk is similar to
\[\left( L^2(\T_{2 \pi})^2, M[\lambda_3] \oplus M[\lambda_4], \dfrac{d}{i dk} \oplus \dfrac{d}{i dk} \right).\]

Let $V_U$ be the intertwining unitary operator from the $4$-state Grover walk $U$ to
\[1 \oplus (-1) \oplus M[\lambda_3] \oplus M[\lambda_4] \in \B(L^2(\T_{2 \pi})^4).\]
Let $V_W$ be the intertwining unitary operator from $W$ to
\[M[\lambda_3] \oplus M[\lambda_4] \in \B(L^2(\T_{2 \pi})^2).\]
Because $\lambda_4$ is not obtained by translation of $\lambda_3$,
the space of uniform operators intertwining $U$ and $W$ is 
\[
V_W^{-1}
\left(
\begin{array}{ccccc}
0 & 0 & C(\T_{2 \pi}) & 0\\
0 & 0 & 0 & C(\T_{2 \pi})\\
\end{array}
\right)
V_U,
\]
where $C(\T_{2 \pi})$ is the space of multiplication operators given by continuous functions on the torus.
The operators in the middle map $(L^2(\T_{2 \pi}))^4$ to $(L^2(\T_{2 \pi}))^2$.
To identify $V_U$ and $V_W$, we only have to identify the eigenspace decomposition of $\widehat{U}(k)$ and $\widehat{W}(k)$.
The calculation is possible, but complicated. We omit identifying them.
\end{example}

\begin{example}\label{example: a direct summand of 3-state Grover walk}
Consider the quantum walk generated by
\[W = 
\frac{1}{3}
\left(
\begin{array}{cc}
-2 - S_1 & \sqrt{2}i (S_1^{-1} - 1)\\
\sqrt{2}i (S_1 - 1) & - 2 - S_1^{-1}
\end{array}
\right) \in \B(\ell_2(\Z) \otimes \C^2).
\]
The characteristic polynomial of the inverse Fourier transform $\widehat{W}(k)$
is
\[\lambda^2 + \dfrac{4 + 2 \cos k}{3} \lambda + 1.\]
The roots are given by $\lambda_2(k)$ and $\lambda_2(k + 2 \pi)$ defined in Example \ref{example: 3-state Grover walk}.
The quantum walk is similar to
\[\left( L^2(\T_\pi), M[\lambda_2], \dfrac{d}{i dk} \right).\]

Let $V_U$ be the intertwining unitary operator from the $3$-state Grover walk $U$ to
\[1 \oplus M[\lambda_2] \in \B(L^2(\T_{2 \pi}) \oplus L^2(\T_{4 \pi})).\]
Let $V_W$ be the intertwining unitary operator from $W$ to
\[M[\lambda_2] \in \B(L^2(\T_{4 \pi})).\]
The space of uniform operators intertwining $U$ and $W$ is 
\[
V_W^{-1}
\left(
0 \quad M(C(\T_{4 \pi}))
\right)
V_U.
\]
The operators in the middle map $L^2(\T_{2 \pi}) \oplus L^2(\T_{4 \pi})$ to $L^2(\T_{4 \pi})$.
\end{example}

In the above two examples, we see the cases that there exist non-zero intertwiners.
However, in many cases, there exists no non-zero uniform intertwiner.
For example, there exists no non-zero uniform intertwiner between the walk in Example \ref{example: two direct summands of 4-state Grover walk} and that in Example \ref{example: a direct summand of 3-state Grover walk}.
To show absence of a non-zero intertwiner, we need some systematic way of proof like the contrapositive of Proposition \ref{proposition: uniform intertwiner}, while
for existence of non-zero intertwiner, we might find intertwiners by some chance.

\begin{example}
Let $\rho$ be a positive real number less than $1$. 
Consider the quantum walks $U_\rho$ on $\ell_2(\Z) \otimes \C^2$ generated by
$U_\rho = \left(
\begin{array}{cc}
\rho S_1^{-1} & \sqrt{1 - \rho^2} S_1^{-1}\\
\sqrt{1 - \rho^2} S_1 & \rho S_1
\end{array}
\right)$.
The eigenvalue functions of the inverse Fourier transform $\widehat{U_\rho}$ are
\[
\lambda_{\rho, +} = \rho \cos k + i \sqrt{1 - \rho^2 \cos^2 k}, 
\quad 
\lambda_{\rho, +} = \rho \cos k - i \sqrt{1 - \rho^2 \cos^2 k}.
\]
If $\rho(1), \rho(2) \in (0, 1)$ and if $\rho(1) \neq \rho(2)$, then there exists no $\alpha \in \T_{2 \pi}$ satisfying $\lambda_{\rho(2), \pm}(k) = \lambda_{\rho(1), \pm}(k + \alpha)$. 
It follows that there exists no uniform intertwiner between $U_{\rho(1)}$ and $U_{\rho(2)}$.

\end{example}

\subsection{Uniform commutant of a homogeneous analytic quantum walk}

For a quantum walk $(\Hil, (U^t)_{t \in \Z}, D)$,
we call the algebra
\[\left\{ V \in \B(\Hil) \left| 
V U = U V, 
k \mapsto \exp(i k D) V \exp(- i k D) {\rm \ is\ continuous} \right. \right\}\]
the {\it uniform commutant} of $U$.
The following are conclusions of Proposition \ref{proposition: uniform intertwiner}.

\begin{corollary}
The uniform commutant of
a prime model quantum walk 
$( \ell_2( r \Z )$, $( U_\lambda^{t} )_{t \in \Z}$, $D_r )$
is identical to
\[\left\{\left. \F_r M[\rho] \F_r^{-1} \ \right|\ \rho \colon \T_{2 \pi r^{-1}} \to \C {\rm\ continuous} \right\}.\]
\end{corollary}

\begin{proof}
By the definition of a prime model quantum walk,
$\lambda$ has no rotational symmetry.
\end{proof}

The following is the motivation of the definition of {\it prime} model quantum walk.

\begin{corollary}
No prime model quantum walk is similar to 
a direct sum of two (not necessarily homogeneous) one-dimensional uniform quantum walks.
\end{corollary}

\begin{proof}
For every prime model quantum walk,
the uniform commutant is
\[\left\{\left. \F_r M[\rho] \F_r^{-1} \ \right|\ \rho \colon \T_{2 \pi r^{-1}} \to \C {\rm\ continuous} \right\} \subset \B(\ell_2(r \Z)).\]
The set of all the orthogonal projections in this algebra is $\{0, {\rm id}\}$.
If the walk were a direct sum of two quantum walks, the set would contain a non-trivial projection.
\end{proof}

\begin{remark}
In \cite{SaigoSako}, the notion of indecomposable quantum walk is defined.
The condition of indecomposable model quantum walk is weaker than that of primeness.
Primeness means that the walk can not be decomposable in the category of quantum walks,
while indecomposability means that the walk can not be decomposable in the category of {\it homogeneous} quantum walks.
\end{remark}

Using Theorem \ref{theorem: structure theorem}, 
we identify the structure of the uniform commutant 
of the discrete-time homogeneous analytic quantum walk
$(\ell_2(\Z) \otimes \C^n, (U^t)_{t \in \Z}, D \otimes {\rm id})$.

\begin{proposition}\label{proposition: uniform commutant}
Let $(\ell_2(\Z) \otimes \C^n, (U^t)_{t \in \Z}, D_1 \otimes {\rm id})$ be an arbitrary one-dimensional discrete-time homogeneous analytic quantum walk.
The uniform commutant
\[
\left\{
V \colon 
\ell_2(\Z) \otimes \C^n \to \ell_2(\Z) \otimes \C^n \ |\
V U = U V, k \mapsto e^{ik D} V e^{-ikD} {\rm \ is\ continuous} 
\right\}
\]
is isomorphic to the algebra of the following form
\[
\bigoplus_{j = 1}^l 
\left(
C (\T_{r(j)}) \otimes M_{\mu(j)}(\C) 
\right)
\oplus
\bigoplus_{k = 1}^m 
\left(
{\rm C}^{*}_{\rm u} (\Z) \otimes M_{\nu(k)}(\C)
\right)
.\]
(The integers $l$ and $m$ can be zero. In such a case, erase the corresponding direct summand.)
The operator $U$ itself is located at the element of the form
\[
\bigoplus_{j = 1}^l 
\left(
M[\lambda_j] \otimes {\rm id}_{\mu(j)} 
\right)
\oplus
\bigoplus_{k = 1}^m 
\left(
\alpha(k) \otimes {\rm id}_{\nu(k)} 
\right).
\]
The structure of smoothness and analyticity is given by the self-adjoint operator
\[
\bigoplus_{j = 1}^l 
\left(
\frac{d}{i d k} \otimes {\rm id}_{\mu(j)}
\right)
\oplus
\bigoplus_{k = 1}^m 
\left(
D_1 \otimes {\rm id}_{\nu(k)}
\right)
.\]
\end{proposition}

\begin{remark}
The non-negative integers $l$ and $m$ in Proposition \ref{proposition: uniform commutant} can be different from those in Theorem
\ref{theorem: structure theorem}.
\end{remark}

\begin{proof}
Recall that $U$ can be decomposed into prime model quantum walks and constant quantum walks.
Let $U_1$ and $U_2$ be two direct summand of $U$.
By the argument in Cases {\bf 1}, {\bf 2}, {\bf 3} in the previous subsection,
if there exists a non-zero uniform intertwiner between $U_1$ and $U_2$ ,
then there exists a uniform unitary operator which intertwines them.
Therefore, existence of non-zero uniform intertwiner defines an equivalence relation between
direct summand of $U$.
Let $\{U_1, U_2, \cdots, U_\nu\}$ be such an equivalence class.
By the above argument, they are all constant quantum walks, or they are all prime model quantum walks.

Consider the case that
$U_1, U_2, \cdots, U_\nu$ are the constant walks.
By {\bf Case 2} in the previous subsection, these are identical.
Express them as $\left( \ell_2(\Z), (\alpha^t)_t, D_1 \right)$.
Combining the set of uniform intertwiners, we obtain the algebra
${\rm C}^*_{\rm u}(\Z) \otimes M_\nu(\C)$.

Consider the case that
$U_1, U_2, \cdots, U_\nu$ are prime quantum walks.
By {\bf Case 3} in the previous subsection, corresponding analytic functions 
\[\lambda_1, \cdots, \lambda_\nu \colon \T_{2 \pi r^{-1}} \to \T\]
are mutually translations of each other.
There exist $\alpha_1, \cdots, \alpha_\nu \in \T_{2 \pi r^{-1}}$ such that
$\lambda_j(k) = \lambda_1 (k + \alpha_j)$.
Then we have $\lambda_j(k) = \lambda_l(k + \alpha_j - \alpha_l)$.
By {\bf Case 3} in the previous subsection, 
the set of uniform intertwiners from $U_l$ to $U_j$ is
\[\left\{\left. \F_r M[\rho] \F_r^{-1} \exp (i (\alpha_j - \alpha_l) D_r) \ \right|\ \rho \colon \T_{2 \pi r^{-1}} \to \C {\rm\ continuous} \right\}.\]
In the case of $l = k$,
$U_l$ is located at $\F_r M[\lambda_l] \F_r^{-1}$.
The inverse Fourier transform is
\[\left\{\left. 
\exp \left(\alpha_j \dfrac{d}{d k} \right) 
M[\rho] 
\exp \left(\alpha_l \dfrac{d}{dk} \right)^{-1} 
\ \right|\ 
\rho \colon \T_{2 \pi r^{-1}} \to \C {\rm\ continuous} \right\}.\] 
Note that the operator $\exp \left(\alpha_j \frac{d}{d k} \right) \in \B(L^2(\T_{2 \pi r^{-1}}))$
is the translation operator by $\alpha_j \in \T_{2 \pi r^{-1}}$ and that it is a normalizer of the space of multiplication operators $C(\T_{2 \pi r^{-1}})$.
Also note that $\exp \left(\alpha_j \frac{d}{d k} \right)$ commutes with the differential operator $\frac{d}{i d k}$.

Combining all the intertwiners, we conclude that
the set of uniform intertwiners between $U_1 \oplus \cdots \oplus U_\nu$
and itself
is isomorphic to
$C(\T_{2 \pi r^{-1}}) \otimes M_\nu(\C)$.
\end{proof}

\section{Realization by a continuous-time uniform quantum walk}

\begin{lemma}\label{lemma: realization by one-parameter}
Let $\nu$ be a natural number.
Let $r$ be a positive real number.
Let $\lambda \colon \T_{2 \pi r^{-1}} \to \T$ be a continuous map.
There exists a one-parameter group $(U^{(t)})_{t \in \R}$ of unitary operators in 
$C (\T_{2 \pi r^{-1}}) \otimes M_{\nu}(\C)$ satisfying
\[U^{(1)} = M[\lambda] \otimes {\rm id}_{\nu} ,\]
if and only if the winding number of $\lambda$ is zero.
\end{lemma}

\begin{proof}
Suppose that the winding number of $\lambda$ is zero.
Then there exists a continuous function $h \colon \T_{2 \pi r^{-1}} \to \R$ such that $\exp(i h) = \lambda$.
The one-parameter unitary group
\[U^{(t)} = M[\exp(i t h)] \otimes {\rm id}_\nu \in C (\T_{2 \pi r^{-1}}) \otimes M_{\nu}(\C)\]
satisfies $U^{(1)} = M[\lambda] \otimes {\rm id}_{\nu}$.

Conversely suppose that there exists a one-parameter unitary group $U^{(t)} \in C (\T_{2 \pi r^{-1}}) \otimes M_{\nu}(\C)$ which satisfies $U^{(1)} = M[\lambda] \otimes {\rm id}_{\nu}$.
We make use of $C (\T_{2 \pi r^{-1}})$-valued determinant
\[\det \colon C (\T_{2 \pi r^{-1}}) \otimes M_{\nu}(\C) \to C (\T_{2 \pi r^{-1}}).\]
Since the map $\det$ is multiplicative, $\det U^{(t)}$ is a unitary element of $C (\T_{2 \pi r^{-1}})$.
The winding numbers $\{w(\det U^{(t)})\}_{t \in \R}$ define a group homomorphism from $\R$ to $\Z$.
It follows that $w(\det U^{(t)}) = 0$ for every $t \in \R$.
Therefore we have
\begin{eqnarray*}
\nu w(\lambda) = w(\lambda^\nu) = w(\det U^{(1)}) = 0
\end{eqnarray*}
and $w(\lambda) = 0$.
\end{proof}

\begin{theorem}\label{theorem: continuous-time}
Let $(\ell_2(\Z) \otimes \C^n, (U^t)_{t \in \Z}, D_1 \otimes {\rm id})$ be an arbitrary one-dimensional discrete-time homogeneous analytic quantum walk.
Let $\lambda_1, \lambda_2, \cdots$ be the eigenvalue functions of $U$
introduced in Subsection \ref{subsection: review of SS}.
Then the following conditions are equivalent
\begin{enumerate}
\item
There exists a one-dimensional continuous-time {\rm uniform} quantum walk $(\ell_2(\Z) \otimes \C^n, (U^{(t)})_{t \in \R}, D_1 \otimes {\rm id})$ such that
$U^{(1)} = U$.
\item
There exists a one-dimensional continuous-time {\rm homogeneous} and {\rm analytic} quantum walk $(\ell_2(\Z) \otimes \C^n, (U^{(t)})_{t \in \R}, D \otimes {\rm id})$ such that
$U^{(1)} = U$.
\item
All the winding numbers of  $\lambda_1, \lambda_2, \cdots $ are zero.
\end{enumerate}
\end{theorem}

The first item looks much weaker than the second item, 
but the following proof will show that both are equivalent to the third item.

\begin{remark}
We may further weaken the first condition.
We can eliminate the assumption that $(U^{(t)})_{t \in \R}$ is continuous with respect to the strong operator topology.
\end{remark}

\begin{proof}
The easier half of
Theorem 5.14 in \cite{SaigoSako} shows that the third condition implies the second one.
It suffices to show that the first condition implies the third one.

Suppose that there exists a one-parameter group $U^{(t)}$ of uniform unitary operator on $\ell_2(\Z) \otimes \C^n$ such that $U^{(1)} = U$.
Note that for every $t$, $U^{(t)}$ commutes with $U$.
We use the algebra
\[
\bigoplus_{j = 1}^l 
\left(
C (\T_{r(j)}) \otimes M_{\mu(j)}(\C) 
\right)
\oplus
\bigoplus_{k = 1}^m 
\left(
{\rm C}^{*}_{\rm u} (\Z) \otimes M_{\nu(k)}(\C)
\right)
\]
in Proposition \ref{proposition: uniform commutant}.
Existence of $U^{(t)}$ means that
\[
\widehat{U} = \bigoplus_{j = 1}^l 
\left(
M[\lambda_j] \otimes {\rm id}_{\mu(j)} 
\right)
\oplus
\bigoplus_{k = 1}^m 
\left(
\alpha(k) \otimes {\rm id}_{\nu(k)} 
\right).
\]
can be realized by a one-parameter group of unitary operators in the algebra.

The winding numbers of constant functions $\alpha(k)$ are $0$, so the latter summand does not have to do with our problem.
We can concentrate on the operator
\[
M[\lambda_j] \otimes {\rm id}_{\mu(j)} 
\in
C (\T_{r(j)}) \otimes M_{\mu(j)}(\C).
\]
By Lemma \ref{lemma: realization by one-parameter},
if this is realized by a one-parameter unitary group inside $C (\T_{r(j)}) \otimes M_{\mu(j)}(\C)$,
then the winding number of $\lambda_j$ is $0$.
\end{proof}

\begin{example}
Let $r$ be a real number greater than $0$ and less than $1$.
Let us consider the quantum walk
\[
U = 
\left(
\begin{array}{cc}
r S^{-1} & - \sqrt{1 - r^2} S^{-1} \\
\sqrt{1 - r^2} S & r S 
\end{array}
\right).\]
acting on $\ell_2(\mathbb{Z}) \otimes \mathbb{C}^2$.
The weak limit theorem for this walk has been shown in \cite{KonnoJMSJ}.
The characteristic polynomial of the inverse Fourier transform $\widehat{U}(k)$ is
\[f(k; z) = \lambda^2 - r \left( e^{ik} + e^{-ik} \right) \lambda + 1.\]
We express $z$ by $e^{i \theta}$.
The roots are
\begin{eqnarray*}
\lambda_1(k) &=& r \cos k + i \sqrt{1 - r^2 \cos^2 k},\\
\lambda_2(k) &=& r \cos k - i \sqrt{1 - r^2 \cos^2 k}.
\end{eqnarray*}
They are single-valued functions.
The winding numbers are both zero.
By Theorem \ref{theorem: continuous-time},
This can be realized by a continuous-time quantum walk.\qed
\end{example}

\begin{example}\label{example: 3-state Grover walk is continuous-time}
The $3$-state Grover walk in Example \ref{example: 3-state Grover walk} can  be realized 
by a continuous-time analytic quantum walk.
We have obtained the constant eigenvalue function $\lambda_1(k) = 1$ and a multi-valued analytic eigenvalue function $\lambda_2$.
The winding number of $\lambda_2 \colon \R / (4 \pi \Z) \to \T$ is zero.
For the same reason, the quantum walk in Example \ref{example: a direct summand of 3-state Grover walk} can be realized by a continuous-time quantum walk.
\end{example}

Even if a homogeneous analytic quantum walk $(U^t)_{t \in \Z}$ is realized by a continuous-time quantum walk $(U^{(t)})_{t \in \R}$, the walk $(U^{(t)})_{t \in \R}$ is not necessarily homogeneous. 

\begin{example}
Let $\beta$ be an element of $\T_{2 \pi} = \R / (2 \pi \Z)$.
Assume that for every integer $x \in \Z$, $x \beta \in \T_{2 \pi}$ is not zero.
Let $\lambda \colon \T_{2 \pi} \to \T$ be an analytic map without period.
Assume that the winding number of $\lambda$ is zero.
Choose an analytic map $g \colon \T_{2 \pi} \to \R$ satisfying $\exp(i g) = \lambda$.
Define $\rho \colon \T_{2 \pi} \to \T$ and $h \colon \T_{2 \pi} \to \R$ by $\rho(k) = \lambda(k + \beta)$ and $h(k) = g(k + \beta)$.
Consider the direct sum of model quantum walks
\[
(\ell_2(\Z), (\F_1 M[\lambda]^t \F_1^{-1})_{t \in \Z}, D_1) 
\oplus
(\ell_2(\Z), (\F_1 M[\rho]^t \F_1^{-1})_{t \in \Z}, D_1).
\]
This is realized by the continuous-time homogeneous quantum walk
\[
U^{(t)} =
\left(
\begin{array}{cc}
\F_1 M[\exp(it g)] \F_1^{-1} & 0\\
0 & \F_1 M[\exp(it h)] \F_1^{-1}\\
\end{array}
\right).
\]
We also consider the one-parameter family of unitary
\[
V^{(t)} =
\left(
\begin{array}{cc}
\cos 2 \pi t & - \sin 2 \pi t \cdot \exp( - i \beta D_1)\\
\sin 2 \pi t \cdot \exp( i \beta D_1) & \cos 2 \pi t\\
\end{array}
\right).
\]
Because the inverse Fourier transform $\exp\left( \beta \frac{d}{dk} \right)$ of the operator $\exp( i \beta D_1)$ is the translation operator on $L^2(\T_{2 \pi})$ by $- \beta \in \T_{2 \pi}$, $V^{(t)}$ commutes with $U^{(t)}$.
It follows that $V^{(t)} U^{(t)}$ is also a one-parameter group of unitary operators and realizes the given quantum walk $\F_1 M[\lambda]^t \F_1^{-1} \oplus \F_1 M[\rho]^t \F_1^{-1}$.
If $t$ is not an element of $\frac{1}{2}\Z$, then $V^{(t)} U^{(t)}$ is not homogeneous.
Moreover, it is not even virtually homogeneous in the sense of Definition \ref{definition: homogeneous}.
\end{example}

Theorem \ref{theorem: continuous-time} provides a powerful way to show that given quantum walk is not a restriction of continuous-time quantum walk.
If one wants to show that given quantum walk can be realized by a continuous-time quantum walk, ad hoc way might be useful, because we might be able to find concrete description.
However, if one wants to show that it can not be realized by a continuous-time quantum walk, ad hoc way can not be useful, and we need some systematic procedure.
The following Corollary gives a sufficient condition for such non-existence.

\begin{corollary}\label{corollary: non-existence}
Let $(\ell_2(\Z) \otimes \C^n, (U^t)_{t \in \Z}, D_1 \otimes {\rm id})$ be an arbitrary one-dimensional discrete-time homogeneous analytic quantum walk.
Denote by $\widehat{U} \in C(\T_{2 \pi}) \otimes M_n(\C)$ the inverse Fourier transform of $U$.
If the winding number of $\det \widehat{U} \colon \T_{2 \pi} \to \T$ is not zero, then
there exists no one-dimensional continuous-time uniform quantum walk $(\ell_2(\Z) \otimes \C^n, (U^{(t)})_{t \in \R}, D_1 \otimes {\rm id})$ such that
$U^{(1)} = U$.
\end{corollary}

\begin{proof}
Let $\lambda_1, \cdots, \lambda_l$ be eigenvalue functions of $U$ introduced in Subsection \ref{subsection: review of SS}.
Whichever the eigenvalue functions are single-valued or multi-valued,
the winding number of  $\det \widehat{U}$ is the sum of the winding numbers of $\lambda_1, \cdots, \lambda_l$.
If  the winding number of $\det \widehat{U} \colon \T_{2 \pi} \to \T$ is not zero, there exists an eigenvalue function whose winding number is not zero.
\end{proof}

\begin{example}
we consider a quantum walk defined by
\[U
=
\left(
\begin{array}{cc}
r S_1 & - b S_1 \\
\overline{b} & r
\end{array}
\right), \quad r \in \mathbb{R}, b \in \mathbb{C}, r^2 + |b|^2 = 1.\]
acting on $\ell_2(\mathbb{Z}) \otimes \mathbb{C}^2$.
Determinant of the Fourier dual is
\[
\det \widehat{U} (k) = \det
\left(
\begin{array}{cc}
r e^{ik} & - b e^{ik} \\
\overline{b} & r
\end{array}
\right)
= e^{ik}.
\]
The winding number of $\det \widehat{U} \colon \T_{2 \pi} \to \T$ is one.
By Corollary \ref{corollary: non-existence}, $U$ can not be realized by a continuous-time uniform quantum walk.
\end{example}

The converse of Corollary \ref{corollary: non-existence} does not hold true.

\begin{example}\label{example: 4-state Grover walk is not continuous-time}
We prove that the $4$-state Grover walk in Example \ref{example: 4-state Grover walk} can not be realized 
by a continuous-time uniform quantum walk.
Determinant of the Fourier dual $\widehat{U}$ is a constant function, so we can not use
Corollary \ref{corollary: non-existence}. 
We obtain four single-valued analytic eigenvalue functions
\[\lambda_1(k) = 1, \quad \lambda_2(k) = -1, 
\quad \lambda_3(k), \quad \lambda_4(k)\]
in Example \ref{example: 4-state Grover walk}.
The winding numbers are
\[w(\lambda_1) = 0, \quad w(\lambda_2) = 0, 
\quad w(\lambda_3) = 1, \quad w(\lambda_4) = -1.\]
By Theorem \ref{theorem: continuous-time}, the quantum walk given by $U$ can not be realized by a continuous-time analytic (not necessarily homogeneous) quantum walk.

By the same reason, the quantum walk in Example \ref{example: two direct summands of 4-state Grover walk} can not be realized by a continuous-time quantum walk.
\end{example}

\bibliographystyle{amsalpha}
\bibliography{dtct3.bib}

\end{document}